\documentclass[11pt]{article}
\usepackage[margin=1in]{geometry}
\usepackage{amsmath,amssymb,amsthm}
\usepackage{bm}
\usepackage{enumitem}
\usepackage{graphicx}
\usepackage{float}
\usepackage{xcolor}
\usepackage{tikz,pgfplots}
\usetikzlibrary{calc,positioning,arrows.meta}
\pgfplotsset{compat=1.14}
\usepackage{subcaption}
\usepackage[ruled]{algorithm2e}
\usepackage{cleveref}

\newtheorem{theorem}{Theorem}[section]
\newtheorem{lemma}[theorem]{Lemma}
\newtheorem{proposition}[theorem]{Proposition}
\newtheorem{corollary}[theorem]{Corollary}
\theoremstyle{definition}
\newtheorem{definition}[theorem]{Definition}
\newtheorem{example}[theorem]{Example}

\newenvironment{proofsketch}{\begin{proof}[Proof Sketch]}{\end{proof}}
\renewcommand{\P}{\mathbb{P}}
\newcommand{\E}{\mathbb{E}}

\title{Bayesian Conversations}
\author{Renato Paes Leme$^{1}$, Jon Schneider$^{1}$, Heyang Shang$^{2}$, Shuran Zheng$^{2}$\thanks{Authors are listed in alphabetical order.}\\
\small $^{1}$Google Research, New York, NY, USA\\
\small $^{2}$Tsinghua University, Beijing, China\\
\small \texttt{renatoppl@google.com, jschnei@google.com}\\
\small \texttt{shanghy22@mails.tsinghua.edu.cn, shuranzheng@mail.tsinghua.edu.cn}}
\date{}
\begin{document}
\maketitle
\begin{abstract}
We study \emph{Bayesian conversations}, a model of interactive information disclosure between two privately informed agents. Unlike Bayesian persuasion, information flows in both directions: agents alternately send messages, with each message depending only on the sender's private type and the public transcript. We study welfare-maximizing conversations subject to ex-post individual rationality, focusing on settings with a single action-taker.

When both agents have binary types, we give a polynomial-time algorithm for computing the optimal $T$-round Bayesian conversation. Our approach combines alternating concavification over the belief space with an exact finite discretization whose size is polynomial in the action set and independent of $T$. We further establish a strict separation between Bayesian conversations and mediator protocols in terms of achievable social welfare under the same ex-post IR requirement. Finally, in the binary-type and binary-action setting, we prove a sharp dichotomy: optimal welfare is either attained within six rounds or approached only asymptotically over infinitely many rounds, with linear convergence. Both regimes can occur and can be distinguished in polynomial time.
\end{abstract}
\section{Introduction}

Strategic communication is central to both computer science and economics.
Bayesian persuasion provides its canonical committed, one-way model
\cite{kamenica2011bayesian}, and its computational aspects have generated a
substantial algorithmic literature \cite{dughmi2016algorithmic}. We instead
study settings that resemble conversations: two agents speak back and forth,
and each message can depend on the sender's private information and the
transcript so far. Negotiations, interviews, and adversarial deliberation all
have this interactive feature, and the information revealed during the
conversation affects the parties' eventual payoffs.

We call our framework \emph{Bayesian conversations}. Interactive communication
also appears in communication and information complexity, secure multiparty
computation, cheap talk, and bargaining. Our design question is different: an
external designer proposes the entire protocol, the agents decide whether to
participate, and participating agents commit to its execution. Thus the model
combines two-sided, sequential disclosure with the commitment familiar from
Bayesian persuasion. It asks what a designed decentralized conversation can
achieve, rather than which communication strategies arise endogenously in an
equilibrium.

Formally, our setting consists of two agents (Alice and Bob), each of whom has a
private type drawn from a known prior distribution: $\theta_A \sim
    \P(\theta_A)$ and $\theta_B \sim \P(\theta_B)$. Our central object of study
is an \emph{unmediated Bayesian conversation}: a finite sequence of $T$
alternating messages. In round $t$, the designated speaker randomizes over a
finite message set as a function only of her own type and the transcript of the
first $t-1$ messages. Each message can therefore refine the listener's belief
about the speaker while leaving the other belief fixed. Geometrically, a
conversation repeatedly splits the initial pair
$(\P(\theta_A),\P(\theta_B))$ into posterior pairs $(q_A,q_B)$, one coordinate
at a time. See \Cref{ex:bilateral_trade} for a concrete application.

When exploring the potential of Bayesian conversations, two fundamental
questions arise: What optimal outcomes can Bayesian conversations
achieve? And how powerful are they as communication protocols? Their
multiple rounds of signaling make both questions challenging.

We study the power of communication protocols in terms of the eventual outcomes
that can be reached. To that end, we will assume that a Bayesian game will be
played after the conversation. In order to avoid issues like equilibrium
selection, we will further assume that only one player (Alice) has a
non-trivial action, and the payoffs depend on three things: (i) Alice's chosen
action; (ii) Alice's type; and (iii) Bob's type. In this scenario, Alice's
chosen strategy depends on both her type and her belief about Bob's type. Given
that her belief about his type is a function of the communication protocol, the
conversation will directly affect the payoff of both parties.

A protocol is designed by an external party and proposed to the agents. Since
Alice is the only agent with a non-trivial action, she always weakly prefers to
participate in the communication protocol -- after all, she can always ignore
the communication and choose an action based on her prior. Bob, however, needs
to be incentivized to participate in the protocol. This will take the form of
an ex-post IR constraint: for every realization of Bob’s type and the terminal transcript of the
communication protocol, Bob prefers the protocol outcome to the outcome when
Alice plays only based on her type and her prior belief about Bob's type (i.e.,
when no conversation or signaling occurs). Our choice of ex-post IR as a participation constraint is natural because it ensures that, after the communication outcome is realized, no type of Bob regrets participating relative to no communication. And it also leads to a rich mathematical structure.

\subsection{Our Results}

In our first main result, we provide a polynomial-time algorithm for computing the
social-welfare-maximizing $T$-round Bayesian conversation when Alice and Bob have
binary types (Theorem~\ref{thm:algo_correctness}). This algorithm is a dynamic program based on the
structural observation that the optimal social welfare of any $T$-round Bayesian
conversation, as a function of the starting prior, is determined by a fixed grid
of size polynomial in the game and independent of $T$. In the binary-type setting, the algorithm allows us to establish a strict separation between mediator protocols and
Bayesian conversations; we construct an example where the maximal social welfare
achievable by an ex-post IR mediator protocol is strictly larger than the maximal
social welfare achievable by a conversation (Theorem~\ref{thm:expost_gap}).

We also consider the value of additional rounds of communication in Bayesian
conversations. Do multiple rounds of interaction help when designing Bayesian
conversations? When both agents have binary types and Alice has binary actions,
we prove a sharp dichotomy: either the optimum is attained within six rounds, or
it is approached only asymptotically, with linear convergence over infinitely
many rounds, and we show that both regimes can occur  (Theorem~\ref{thm:converge_BC}). The proof reduces
the dynamic program to the evolution of a small interior grid and classifies all
possible update cycles; it also yields a polynomial-time test for the two regimes.

\subsection{Related Work}
\paragraph{Two-sided interactive communication.}
One close precedent is Mao et al.~\cite{mao2022interactive}. Like us, they study two
privately informed agents who communicate
before a game with one nontrivial action-taker: each round refines one belief
coordinate, and alternating concavification is used to study round complexity
and limiting behavior. Their communication is generated endogenously by
equilibrium and voluntary information-sharing incentives, whereas we optimize
an externally designed protocol subject to ex-post participation.
Our exact finite-grid result for binary types, mediator-versus-conversation
welfare comparison, and convergence classification address different questions.

Crawford and Sobel~\cite{crawford1982strategic} and Farrell and
Rabin~\cite{ferrell1996cheap} study cheap talk. Krishna and
Morgan~\cite{krishna2004art} show how multi-stage communication can improve
information transmission. Amitai~\cite{amitai1996cheap} studies two-sided
incomplete information, and Chiba and Hori~\cite{chiba2022two} analyze two-way
communication between a privately informed decision maker and expert.
Rosenberg et al.~\cite{rosenberg2013strategic} study reciprocal information
exchange in repeated games. These works analyze equilibrium disclosure and focus on endogenous strategic information exchange;
we instead treat the communication protocol itself as the object of design and ask what outcomes can be achieved by the protocol.
Hart~\cite{hart1985nonzero} and Aumann and
Hart~\cite{aumann1986bi,aumann2003long} use bimartingales to characterize
equilibrium communication, while Koessler et al.~\cite{Koessler2022long} study
repeated disclosure by competing designers.

\paragraph{Mediated versus unmediated communication.}
Goltsman et al.~\cite{goltsman2009mediation} compare negotiation, mediation,
and arbitration; Arieli et al.~\cite{arieli2023value} study the value of
mediation over long cheap talk. Corrao and Dai~\cite{corrao2023bounds} compare
mediation, cheap talk, and persuasion under transparent motives. These models
have an informed sender and an uninformed receiver; we compare welfare under
committed protocols with two-sided private information and ex-post participation.
Abraham et al.~\cite{abraham2006distributed,abraham2008lower} ask when a mediated equilibrium
can be implemented by direct communication, emphasizing resilience, robustness,
and communication requirements. Our comparison has a different
target: we do not implement a fixed mediated outcome, but compare the maximum
welfare achievable under the two communication architectures. 
More broadly,
interactive communication underlies communication and information complexity
and secure multiparty computation
\cite{andrew1979some,kushilevitz1997communication,braverman2012interactive,basu2022geometry};
our agents are strategic and their participation constraints restrict feasible
protocols.

\paragraph{Persuasion and richer information structures.}
Kamenica and Gentzkow~\cite{kamenica2011bayesian} study committed disclosure by an informed sender
to an action-taking receiver; its strategic and
computational foundations are well developed
\cite{bergemann2016information,dughmi2016persuasion}. Extensions consider
privately informed receivers \cite{kolotilin2017persuasion},
multiple senders \cite{gentzkow2016competition}, multiple receivers
\cite{arieli2021feasible}, and intermediaries
\cite{mahzoon2022hierarchical,arieli2022bayesian}. Babichenko et
al.~\cite{babichenko2021multi} study multi-channel and networked communication
structures, whereas Brooks et al.~\cite{brooks2022information} study
information hierarchies, including organizational information structures.
Zhang et al.~\cite{zhang2026expost} impose ex-post IR in one-way persuasion:
the sender must weakly prefer each realized disclosure to no disclosure.
We study ex-post participation in alternating exchanges between two privately
informed agents, with welfare and round complexity as our objectives.

\paragraph{Geometric approaches to interactive protocols.}
A series of papers by Basu et
al.~\cite{basu2022geometry,basu2023randomized,basu2025linear}
develops a geometric theory of two-party secure computation based on recursive
convexification along one coordinate at a time. They use this perspective to
characterize secure realizability, round complexity, and associated lamination
hulls. Our recursive technique is geometrically similar,
but our model and objective are different: we maximize welfare subject to
ex-post participation constraints and ask whether alternating concavification
attains its limiting value in finitely many rounds.

\section{Problem Description and Preliminaries}
We consider a two-player setting where both players (Alice and Bob) have
private types, denoted by $\theta_A$ and $\theta_B$ respectively. We assume
that their types are drawn independently from commonly known distributions with
finite support, i.e., $\theta_A$ and $\theta_B$ are drawn from $\P(\theta_A)
	\in \Delta (\Theta_A)$ and $\P(\theta_B) \in \Delta (\Theta_B)$ independently,
where $|\Theta_A|,|\Theta_B|$ are finite and these prior distributions are
common knowledge.

\subsection{Communication protocols}

Before delving into the game-theoretic aspects of this problem, we begin by
establishing the communication models via which Alice and Bob can share
information about their types. In general, we will consider settings where
Alice and Bob agree on the precise protocol of communication ahead of time;
later, we will impose constraints on this protocol so that Alice and Bob are
incentivized to participate in the protocol and not defect.

Broadly, we consider two different classes of protocols. In the first, we
assume Alice and Bob alternate revealing information about their private types
to one another via a pre-determined protocol. This could represent Alice and
Bob communicating via a platform which sets the communication rules and
monitors the communication process. %

\begin{definition}[Bayesian conversations]\label{def:two-way}
	For two players with private types $\theta_A$ and $\theta_B$, a $T$-round
	Bayesian conversation is a tuple
	$\pi=\langle (M_t,\sigma_t,\varphi_t)\rangle_{t=1}^T$. Here $M_t$ is a finite
	message space, $\sigma_t\in\{A,B\}$ is the speaker in round $t$, and speakers
	alternate: $\sigma_{t+1}\ne\sigma_t$. Writing
	$h^{(t-1)}=(m_1,\ldots,m_{t-1})$ for the public transcript, the signaling rule
	\[
		\varphi_t:\Theta_{\sigma_t}\times
		\prod_{j<t}M_j\longrightarrow\Delta(M_t)
	\]
	maps only the speaker's type and the observed transcript to a distribution over
	messages. Thus, conditional on $\theta_A,\theta_B$, the message in round $t$ is
	drawn as $m_t\sim\varphi_t(\theta_{\sigma_t},h^{(t-1)})$. Since Alice takes the payoff-relevant action after communication,
we adopt the convention that Bob sends the final message. In particular, Bob speaks first when $T$ is odd, and Alice
speaks first when $T$ is even.
\end{definition}
We denote by $\Pi_{\mathrm{BC}}$ the class of all Bayesian conversations. One
round therefore means one message, a convention used throughout the paper. We
focus on full participation: once both agents join, they must complete the
protocol.

\begin{definition}[Committed Bayesian Conversation protocols]
	A committed Bayesian Conversation protocol $\pi$ is a protocol in which Alice and Bob decide at the beginning whether to communicate via a Bayesian conversation $\pi$ or not, and once they agree, they must complete the full protocol without quitting.
\end{definition}

The joint distribution over types and a complete transcript is
\begin{align*}
	\P_\pi(\theta_A,\theta_B,h^{(T)})
	=\P(\theta_A)\P(\theta_B)
	\prod_{t=1}^T\varphi_t(m_t\mid\theta_{\sigma_t},h^{(t-1)}).
\end{align*}
After round $T$, Alice's belief about $\theta_B$ is
$q_B^{(T)}=\P_\pi(\theta_B\mid\theta_A,h^{(T)})$ and Bob's belief about
$\theta_A$ is $q_A^{(T)}=\P_\pi(\theta_A\mid\theta_B,h^{(T)})$. We study the
induced distribution
$\P_\pi(\theta_A,\theta_B,q_B^{(T)},q_A^{(T)})$, called the \emph{joint
	posterior distribution} induced by $\pi$. Note that although $\theta_A, \theta_B$ are the players' private types, the
posterior distribution induced by $\pi$ is commonly known by both players. This
is because we assume the priors $\P(\theta_A)$ and $\P(\theta_B)$ are common
knowledge and the protocol $\pi$ is commonly known.

Second, we consider communication that involves a trusted third party who we
refer to as \textit{the mediator}. In this setting, we consider
``revelation-style'' protocols, where Alice and Bob begin by revealing their
full type to the mediator, who then sends a single public signal to both
players. We assume that the mediator is not strategic and will faithfully
execute the protocol. %
\begin{definition}[Mediator protocols]
	For two players with private type $\theta_A$ and $\theta_B$, a mediator's protocol $\pi$ is a signaling scheme $\langle S, \{\pi(\cdot|\theta_A, \theta_B)\}_{\theta_A\in \Theta_A, \theta_B \in \Theta_B}\rangle$ that specifies how the mediator sends the public signal: when the players' types are $\theta_A$ and $\theta_B$, the mediator sends signal $s\in S$ with probability $\pi(s|\theta_A, \theta_B)$. We denote by $\Pi_{\text{M}}$ the class of all mediator protocols.
\end{definition}

After observing the public signal, the players update their beliefs about the
other player's type according to Bayes' rule. For a signaling scheme $\pi$, let
the joint distribution of the types and the signals be $\P_\pi( \theta_A,
	\theta_B, s)$. Then Alice's posterior about $\theta_B$ after seeing $s$ is $q_B
	= \P_\pi(\theta_B|\theta_A, s)$, and Bob's posterior about $\theta_A$ after
seeing $s$ is $q_A = \P_\pi(\theta_A| \theta_B, s)$. Again, we are interested
in the joint distribution of players' types and beliefs after seeing the public
signal, and we say that the distribution of the posteriors $\P(\theta_A,
	\theta_B, q_B, q_A)$ is the distribution of posteriors induced by $\pi$.

\subsection{Actions and incentives}
\label{sec:actions}

After Alice and Bob communicate, the two players will play a ``game''. As in
the case of Bayesian persuasion, in this paper we primarily focus on the very
simple class of games with a single action-taker. In particular, we assume that
Alice takes an action $r \in R$ (for some finite action set $R$) and that by
taking this action Alice receives utility $u_A(\theta_A, \theta_B, r)$ and Bob
receives utility $u_B(\theta_A, \theta_B, r)$.\footnote{ We briefly note that
our model extends slightly beyond the setting where only one player takes an
action and captures the case where Alice and Bob play a two-player
\emph{Stackelberg game} $G$, where Alice takes an action $r_A$, Bob responds
(after seeing $r_A$) with an action $r_B$, Alice receives utility $G_{A}(r_A,
	r_B, \theta_A)$ and Bob receives utility $G_{B}(r_A, r_B, \theta_B)$ (note that
in contrast to the model above, here we assume their utilities in this game
only depend on the actions taken and their individual private types). If we let
$r_B(r_A, \theta_B) = \arg\max_{r_B}G_{B}(r_A, r_B, \theta_B)$ (i.e., $r_B$ is
the best response for Bob to $r_A$), then our model captures this by letting
$u_{A}(\theta_A, \theta_B, r) = G_{A}(r, r_B(r, \theta_B), \theta_A)$ and
$u_B(\theta_A, \theta_B, r) = G_{B}(r, r_B(r, \theta_B), \theta_B)$. } Then
Alice's optimal strategy in such a setting is clear: if she has a posterior
belief $q_B$ for Bob's type, Alice will take the action that maximizes her
expected utility given this posterior, namely%
\begin{align*}
	r^*(\theta_A, q_B) = \arg\max_{r \in R} \E_{\theta_B\sim q_B} u_A(\theta_A, \theta_B, r).
\end{align*}

\begin{example}[Bilateral trade]\label{ex:bilateral_trade}
	One specific Stackelberg game we use as a running example is \emph{bilateral trade}. Alice and Bob's types belong to finite subsets $\Theta_A,\Theta_B\subseteq[0,1]$ and represent their valuations for an item that Alice may sell to Bob. After communicating, Alice sets a price $r \in [0, 1]$, which Bob accepts iff $r \leq \theta_B$. In terms of the notation above,
	\[
		u_A(\theta_A,\theta_B,r)=(r-\theta_A)\bm{1}[r\le \theta_B],
		\qquad
		u_B(\theta_A,\theta_B,r)=(\theta_B-r)\bm{1}[r\le \theta_B].
	\]
\end{example}

Alice's optimal action may change depending on the communication between the
two players, and this communication may modify both players' expected utility.
In particular, Bob can easily find himself in a situation where he would prefer
not to share a specific piece of information with Alice (e.g., in the bilateral
trade game, if Bob fully reveals his type $\theta_B$, Alice's optimal strategy
is to set a price of $\theta_B$ which leaves Bob with no net utility).

Since Alice is the sole action taker, she always (weakly) benefits from communication. Thus only Bob's participation constraint can bind.

\begin{definition}[Ex-post IR for committed protocols]
	\label{def:expost_IR}
	For each type of Alice, define her no-communication action by
	\[
		r^0(\theta_A)
		=
		\arg\max_{r\in R}
		\mathbb{E}_{\theta_B\sim \mathbb P(\theta_B)}
		\bigl[u_A(\theta_A,\theta_B,r)\bigr].
	\]

	A committed Bayesian conversation $\pi$ is \emph{ex-post individually rational}
	if, for every $\theta_B\in\Theta_B$ and every terminal history $h$ such that
	$\mathbb P_\pi(\theta_B,h^{(T)}=h)>0$, we have
	\begin{align}
		\mathbb{E}_{\pi}\!\left[
		u_B\!\left(
		\theta_A,\theta_B,
		r^*(\theta_A,q_B^{(T)})
		\right)
		\,\middle|\,
		\theta_B,h^{(T)}=h
		\right]
		\ge
		\mathbb{E}_{\theta_A\sim \mathbb P(\theta_A)}
		\left[
			u_B\!\left(
			\theta_A,\theta_B,r^0(\theta_A)
			\right)
			\right].
		\label{eqn:expost_IR}
	\end{align}

	For a mediator protocol, the same inequality is required after every public
	signal $s$ with $\mathbb P_\pi(\theta_B,s)>0$, replacing the terminal history
	by $s$.
\end{definition}

For the ex-post IR mediator protocol, we have an exponential-size LP algorithm
that can compute the protocol achieving optimal social welfare:

\begin{proposition}\label{thm:ex-post-lp}
	The optimal ex-post IR mediator protocol that maximizes the expectation of a utility function $u(\theta_A, \theta_B,r)$ can be solved by a linear program with size $O(|\Theta_A|\cdot|\Theta_B|\cdot|R|^{|\Theta_A|})$.
\end{proposition}

We defer the proof of this proposition to Appendix~A.

\section{Optimal ex-post IR Bayesian conversations}
In contrast to the case of mediator protocols, finding a Bayesian conversation
that achieves optimal social welfare is significantly more challenging due to
multiple rounds of communication. Our main algorithmic result is the following.

\begin{theorem}\label{thm:algo_correctness}
    If Alice and Bob each have two types, then for every $T$ there is an
    algorithm running in time $\operatorname{poly}(T,|R|)$ that computes the
    highest welfare achievable by an ex-post IR $T$-round Bayesian conversation
    ($\lceil T/2\rceil$ messages by Bob and $\lfloor T/2\rfloor$ by Alice).
\end{theorem}

We prove the theorem constructively. The main challenge is that a conversation
can generate continuously many posterior beliefs; the key is an exact finite
grid that remains sufficient after every round of alternating concavification.

\subsection{Algorithm for finding the optimal ex-post IR Bayesian Conversation}

\label{sec:algorithm}
We will focus on the case where each agent has two types $\Theta_A = \Theta_B = \{L,H\}$ and we aim to maximize the expected social welfare. We will refer to the prior that agents have about each other's types before any communication as $q_A^0, q_B^0$. We will also assume that Alice's choice of action $r^*(\theta_A, q_B)$ breaks ties in favor of Bob, i.e., between two actions that lead to the same payoff for her, she chooses the best action for Bob, which is also the action that maximizes social welfare.

In that case, the beliefs of Alice and Bob can be described by a pair $(q_B,
    q_A) \in [0,1]^2$ representing the belief that the other agent has type $H$.

Overall, the algorithm follows a dynamic programming procedure from the final
round to the first round: At the final round, for any posterior belief $(q_B,
    q_A)$ that can be induced by a Bayesian conversation, we can compute the
corresponding social welfare $W_0(q_B, q_A)$.

During each round of interaction, Alice and Bob effectively refine the belief
about her/himself $q$ into new beliefs $q'$. Consequently, the optimal social
welfare attainable at a belief in the previous round can be derived from the
concave hull of the values in the next round along the $x$- or $y$-direction
(\Cref{eq:wel_opt_x}, \Cref{eq:wel_opt_y}). Hence, starting from $W_0$ (the
welfare at the final round) and proceeding backward, we can thus compute the
optimal social welfare at any belief and any earlier round via dynamic
programming.

Note that the ex-post IR condition (Equation~\eqref{eqn:expost_IR}) is a
condition on the state of the agent's beliefs after the last round of the game.
To establish the ex-post IR property, we first define the ex-post IR region,
based on which the algorithm is initialized. This guarantees that the resulting
outcome of our algorithm satisfies Bob’s ex-post IR constraint in the final
round.
\paragraph{Ex-post IR region}

The region of the space of beliefs $[0,1]^2$ that lead to ex-post IR in the
last period is given by: \newcommand{\IR}{\mathsf{IR}}
\begin{align*}
    \IR_0 = \{ & (q_B, q_A) \in [0,1]^2 \mid  \forall \theta_B \in \text{supp}(q_B), \\                                &
                                     \E_{\theta_A\sim q_A}u_B(\theta_A,\theta_B,r^*(\theta_A,q_B)) \ge \E_{\theta_A\sim q_A^0}u_B(\theta_A,\theta_B,r^0(\theta_A))\}
\end{align*}
where $ \text{supp}(q_B) \subseteq \Theta_B$ is the set of types that occur with non-zero probability in $q_B$ (also, recall that $q_A^0$ is Bob's prior belief about Alice's type). Using this notion we can define the starting point of the recursion as:
\begin{align*}
    W_0(q_B, q_A) =
    \begin{cases}
        W^*(q_B, q_A) & \text{ if } (q_B, q_A) \in \IR_0 \\-\infty &\text{ o.w.}
    \end{cases}
\end{align*}

\noindent where $W^*$ is the optimal welfare given the final beliefs $(q_A, q_B)$ without
taking IR into account. Notice that since Alice breaks ties in favor of Bob, it
has the following form:
\begin{align*}
    W^*(q_B, q_A) =\E_{\theta_B,\theta_A \sim (q_B,q_A)} [u_A(\theta_{A},\theta_{B},r^*(\theta_A, q_B))+u_B(\theta_{A},\theta_{B},r^*(\theta_A, q_B))]
\end{align*}

\paragraph{Recursive Step}
Assuming we know $W_{k-1}$, we can compute $W_{k}$ by a process of alternating
concavification: i.e, we will take concave hulls of the function with respect
to Alice's belief (for odd $k$) and Bob's belief (for even $k$). 
By the standard splitting lemma
\cite{blackwell1953equivalent,aumann1995repeated}, we can view each round of
communication by, e.g., Alice as splitting the belief $q_A$ into a
mean-preserving distribution of posterior beliefs $Q'_A$ supported on
$[0,1]$.
From this perspective, we will update:
\begin{align}
    W_k(q_B,q_A)
     & = \max_{{Q'_B:
             \E_{q'_B\sim Q'_B}[q'_B]=q_B}}
        \E_{q'_B\sim Q'_B}[W_{k-1}(q'_B,q_A)]
        \quad \text{for odd }k
    \label{eq:wel_opt_x}                            \\
    W_k(q_B,q_A)
     & = \max_{{Q'_A:
             \E_{q'_A\sim Q'_A}[q'_A]=q_A}}
        \E_{q'_A\sim Q'_A}[W_{k-1}(q_B,q'_A)]
        \quad \text{for even }k
    \label{eq:wel_opt_y}
\end{align}

Such functions provide enough information to compute the optimal social
welfare:

\begin{lemma}\label{lem:correctness} Given the $W_k(q_B, q_A)$ functions defined above, the optimal social welfare obtained by a protocol with $k$ rounds is $W_{k}(q_B^0, q_A^0)$, i.e., the function evaluated on the original prior.
\end{lemma}

We defer the proof of this lemma to Appendix~B.1.

These steps essentially implement a dynamic programming procedure, where each
iteration backwardly updates values for earlier rounds based on the results
from later ones.

For example, in a four-round conversation, $W_1(0.6,0.4)$ and
$W_4(0.6,0.4)$ give the optimal welfare from belief
$(q_B,q_A)=(0.6,0.4)$ in rounds~4 and~1, respectively.

\begin{figure}[!htb]
    \centering
    \begin{minipage}[t]{0.42\textwidth}
        \centering
        \begin{tikzpicture}
            \begin{axis}[
                    xmin=-0.03,
                    xmax=1.15,
                    ymin=-0.03,
                    ymax=1.15,
                    axis lines = middle,
                    xlabel = {$q_B(\theta_B = H)$},
                    ylabel = {$q_A(\theta_A = H)$},
                    width= \textwidth,
                    height = 0.9\textwidth,
                    x label style={at={(axis description cs:0.5,-0.15)},anchor=north},
                    y label style={at={(axis description cs:-0.17,.5)},rotate=90,anchor=south},
                    xtick={0, 0.3, 1},
                    xticklabels = {$0$, $q_1$, $1$},
                    ytick={0, 0.8, 1},
                    yticklabels = {$0$, $q_2$, $1$},
                ]
                \fill[gray, fill opacity=0.5] (axis cs:0.3,0.4) rectangle (axis cs:1,1);
                \fill[gray, fill opacity=0.5] (axis cs:0,0) rectangle (axis cs:0.3,0.8);
                \filldraw[blue] (0.55,0.4) circle (2pt) node[anchor=north west]{$(q_B^0, q_A^0)$};
                \node at (axis cs:0.65,0.7) {$\IR_0$};
                \node at (axis cs:0.15,0.4) {$\IR_0$};
            \end{axis}
        \end{tikzpicture}
        \caption{A possible illustration of ex-post IR region $\IR_0$. When $q_1 < q_B < 1$, Alice’s behavior remains unchanged and Bob’s IR holds only if $q_A \ge q_A^0$;
            when $0 < q_B < q_1$, Alice takes different actions and Bob’s IR holds only if $q_A \le q_2$.}
        \label{fig:generic_ir_region}
    \end{minipage}
    \hfill
    \begin{minipage}[t]{0.42\textwidth}
        \centering
        \begin{tikzpicture}
            \begin{axis}[
                    xmin=-0.03,
                    xmax=1.15,
                    ymin=-0.03,
                    ymax=1.15,
                    axis lines = middle,
                    xlabel = {$q_B(\theta_B = H)$},
                    ylabel = {$q_A(\theta_A = H)$},
                    xtick = {0, 0.25, ..., 1},
                    ytick = {0, 0.25, ..., 1},
                    width=\textwidth,
                    height = 0.9\textwidth,
                    x label style={at={(axis description cs:0.5,-0.15)},anchor=north},
                    y label style={at={(axis description cs:-0.23,.5)},rotate=90,anchor=south}
                ]

                \draw[-{Stealth[length=1.6mm,width=1mm]},thick, blue]
                (axis cs:0.5,0.25) -- node[right]{$\frac{1}{2}$} (axis cs:0.5,0.5);
                \draw[-{Stealth[length=1.6mm,width=1mm]},thick, blue]
                (axis cs:0.5,0.75) -- node[right]{$\frac{1}{2}$} (axis cs:0.5,0.5);
                \draw[-{Stealth[length=1.6mm,width=1mm]},thick]
                (axis cs:0,0.25) -- node[below]{$\frac{1}{3}$} (axis cs:0.5,0.25);
                \draw[-{Stealth[length=1.6mm,width=1mm]},thick]
                (axis cs:0.75,0.25) -- node[below]{$\frac {2}{3}$} (axis cs:0.5,0.25);
                \draw[-{Stealth[length=1.6mm,width=1mm]},thick]
                (axis cs:0.3,0.75) -- node[above]{$\frac{2}{3} $} (axis cs:0.5,0.75);
                \draw[-{Stealth[length=1.6mm,width=1mm]},thick]
                (axis cs:0.9,0.75) -- node[above]{$\frac{1}{3}$} (axis cs:0.5,0.75);
                \filldraw[red] (axis cs:0.5,0.5) circle (1pt) node[anchor=east]{$W_2$};
                \filldraw[blue] (axis cs:0.5,0.25) circle (1pt) node[anchor=north]{$W_1$};
                \filldraw[blue] (axis cs:0.5,0.75) circle (1pt)
                node[anchor=south]{$W_1$};
                \filldraw[black] (axis cs:0.3,0.75) circle (1pt) node[anchor=south east]{$W_0$};
                \filldraw[black] (axis cs:0.9,0.75) circle (1pt)    node[anchor=south west]{$W_0$};
                \filldraw[black] (axis cs:0,0.25) circle (1pt)
                node[anchor=north west]{$W_0$};
                \filldraw[black] (axis cs:0.75,0.25) circle (1pt)
                node[anchor=north west]{$W_0$};
            \end{axis}
        \end{tikzpicture}
        \caption{Illustration of the recursive steps: one possible computation path of $W_2(0.5, 0.5)$. The value $W_2$ is obtained from neighboring $W_1$
values, each of which is computed from the corresponding $W_0$
values through the preceding concavification step.}
        \label{fig:dp_recursion}
    \end{minipage}
\end{figure}

\paragraph{Exact discretization}
The recursion above specifies continuous functions $W_k : [0,1]^2 \rightarrow
    \mathbb{R}$. However, directly computing the recursive steps in the continuous
domain is not computationally feasible. Therefore, to address this issue, we
show that it is possible to discretize the space $[0,1]^2$ into finitely many
points and only compute the $W_k$ restricted to this set. A crucial fact here
is that the discretization is exact, i.e., it introduces no approximation
errors. In this way, we transform the continuous recursion into a
computationally tractable discrete algorithm.

\begin{lemma}\label{lemma:bilinear} There are finite sets $X^* = \{x_1 <x_2 < \hdots < x_m\}$ and $Y^* = \{y_1 < y_2 < \hdots < y_n\}$ with size $O(|R|)$ such that $W_0$ restricted to the rectangle $(x_i, x_{i+1}) \times (y_j, y_{j+1})$ is bilinear, where $R$ is the set of Alice's actions.
\end{lemma}

Let $X^*$ consist of the values of $q_B$ at which Alice's optimal action changes.
Then Alice's action is fixed within each interval, so $W^*$ is bilinear there and
$|X^*|=O(|R|)$. Since $W_0$ agrees with $W^*$ on $\IR_0$ and equals $-\infty$
outside $\IR_0$, we collect all endpoints of the corresponding IR intervals into
$Y^*$. This ensures that IR feasibility is unchanged within each cell and hence
that $W_0$ is bilinear (or identically $-\infty$) on every cell. Since each
IR interval contributes at most two endpoints, $|Y^*|=O(|R|)$. The detailed
construction and proof are given in Appendix~B.2.

This cellwise bilinearity also yields the key closure property behind our
algorithm: inductively, every off-grid point between two adjacent grid points is
weakly dominated by the interpolation of their values and can therefore be
replaced by those neighboring grid points in any optimal belief split. Thus it
suffices to track only the grid $X^*\times Y^*$, leading to the discrete
algorithm below.

\begin{lemma} \label{lem:update}

    Given the values of $W_{k-1}(x,y)$ for $(x,y) \in X^* \times Y^*$ we can
    determine $W_{k}(x,y)$ by solving the following optimization programs:
    \begin{equation*}
        W_k(\hat x,\hat y)=
        \begin{cases}
            \displaystyle
            \max_{{p\in\Delta_m:
                    \sum_{i=1}^m p_i x_i=\hat x}}
                \sum_{i=1}^m p_i W_{k-1}(x_i,\hat y),
                \quad \text{for odd } k \\
            \displaystyle
            \max_{{p\in\Delta_n:
                    \sum_{j=1}^n p_j y_j=\hat y}}
                \sum_{j=1}^n p_j W_{k-1}(\hat x,y_j),
                \quad \text{for even } k
        \end{cases}
    \end{equation*}
    where $\Delta_\ell=\{p\in[0,1]^\ell:\sum_{r=1}^\ell p_r=1\}$ denotes the $(\ell-1)$-dimensional simplex.
\end{lemma}

\begin{proofsketch}
    We prove the claim by induction on $k$. Here we assume $k$ is odd. The proof for the even case is entirely similar.

    The main argument of the proof is as follows. When updating $W_k(\hat
        {x},\hat{y})$ for any $\hat {x},\hat{y}$, assume $x_i < \bar x < x_{i+1}$ and
    $\bar x = q x_i + (1 - q) x_{i+1}$. By the induction hypothesis, $W_{k-2}$ is
    linear in $x$ on segments $[x_i, x_{i+1}]\times{y_j}$ for all $y_j$'s,
    combining the definition of $W_{k-1}$, we have
    \begin{align*}
        W_{k-1}(\bar x,\hat y)
         & = \textstyle\sum_{j=1}^n p_j W_{k-2}(\bar x, y_j)      \\
         & = q\textstyle\sum_{j=1}^n p_jW_{k-2}(x_i, y_j)
        +(1-q)\textstyle\sum_{j=1}^n p_jW_{k-2}(x_{i+1}, y_j)     \\
         & \le qW_{k-1}(x_i,\hat y)+(1-q)W_{k-1}(x_{i+1},\hat y).
    \end{align*}
    (Here we assume $k \ge 3$. The cases $k = 1$ and $k = 2$ are similar,
    with only minor technical differences (only requiring the application of
    \Cref{lemma:bilinear}), and are provided in the full proof in
    Appendix~B.3.)

    Hence, when updating $W_k(\hat x, \hat y)$, for any $\bar x \notin X^*$ we do
    not need to consider $W_{k-1}(\bar x,\hat y)$, since its value is (weakly)
    lower than the convex combination $q\,W_{k-1}(x_i,\hat y) + (1 -
        q)\,W_{k-1}(x_{i+1},\hat{y})$.

    Moreover, by the induction hypothesis, assume $y_j < \hat y < y_{j+1}$ and
    $\hat y = q' y_j + (1 - q') y_{j+1}$. Then $W_{k-1}$ is linear in $y$ on both
    segments $\{x_i\} \times [y_j, y_{j+1}]$ and $\{x_{i+1}\} \times [y_j,
            y_{j+1}]$. Consequently,
    \begin{align*}
         & W_{k-1}(x_i,\hat y)
        = q'\,W_{k-1}(x_i,y_j) + (1 - q')\,W_{k-1}(x_i,y_{j+1}),
        \\
         & W_{k-1}(x_{i+1},\hat y)
        = q'\,W_{k-1}(x_{i+1},y_j) + (1 - q')\,W_{k-1}(x_{i+1},y_{j+1}).
    \end{align*}

    Therefore, throughout the recursive step, it suffices to track the values on
    the grid $X^* \times Y^*$ only.

\end{proofsketch}

We observe that while it is useful to describe $W_k$ as a linear program, there
is a linear-time algorithm to compute $W_k(\cdot, \hat y)$ as a function of
$W_{k-1}(\cdot, \hat y)$ as a concave hull computation problem.

We summarize and unify all the preceding steps in \Cref{alg:alg1}.

{
\SetAlgoNoEnd \DontPrintSemicolon
\begin{algorithm}[H]
    \caption{Optimal $T$-round ex-post IR Bayesian conversation}
    \label{alg:alg1}
    \KwIn{Round bound $T$, priors $q_A^0,q_B^0$, utilities $u_A,u_B$}
    \KwOut{Optimal social welfare $W_T(q_B^0, q_A^0)$}
    Compute $\IR_0, X^*, Y^*$\;
    \ForEach{$(x^*, y^*) \in X^*\times Y^*$}{
        Compute $W_0(x^*, y^*)$\;
    }
    \For{$k = 1$ \KwTo $T$}{
        \ForEach{$(x^*, y^*) \in X^*\times Y^*$}{
            Compute $W_k(x^*, y^*)$\;
        }
    }
    \If{$q_B^0 \notin X^*$} {
    Compute $W_T(q_B^0, q_A^0)$:\;
    Find consecutive $x_i,x_{i+1}\in X^*$ such that
    $x_i<q_B^0<x_{i+1}$\;
    Set $q\in(0,1)$ such that $q_B^0=qx_i+(1-q)x_{i+1}$\;
    \If{$T$ is even}{
    \ForEach{$y_j\in Y^*$}{
    \(
    \begin{aligned}
    W_{T-1}(q_B^0,y_j)=qW_{T-1}(x_i,y_j)+(1-q)W_{T-1}(x_{i+1},y_j)
    \end{aligned}
    \)\;
    }
    \(\displaystyle
    W_{T}(q_B^0,q_A^0)=
    \max_{{p\in\Delta_n:
            \sum_{j=1}^n p_jy_j=q_A^0}}
        \sum_{j=1}^n p_jW_{T-1}(q_B^0,y_j)
    \)\;
    }
    \Else{
    \(W_{T}(q_B^0,q_A^0)=qW_{T}(x_i,q_A^0)+(1-q)W_{T}(x_{i+1},q_A^0)\)\;
    }
    }
    \Return{$W_T(q_B^0, q_A^0)$}
\end{algorithm}
}
By storing an optimal belief split at each state and each round,
the dynamic program also gives a compact representation of
an optimal conversation.

Computing each one-dimensional concave hull is linear in the number of grid
points. Since $|X^*|,|Y^*|=O(|R|)$ and the same grid is used in all $T$ rounds,
\Cref{alg:alg1} has the running time stated in \Cref{thm:algo_correctness}. A
formal correctness proof appears in Appendix~B.4.

\paragraph{Beyond binary types.}
The proof of Theorem~\ref{thm:algo_correctness} relies on the exact
discretization result in \Cref{lem:update}: in the binary case, throughout the
alternating concavification process it is enough to update the values only on
the finite set of vertices of the initial bilinear subdivision. A natural
extension of \Cref{alg:alg1} to larger type spaces would apply the same
alternating concavification recursion on higher-dimensional belief simplices
and again keep track only of those original vertices.

However, beyond binary type spaces, this natural extension of the original
algorithm is no longer guaranteed to return the correct value. In
Appendix~B.5, we first provide intuition for why this
extension fails by constructing a simple $2$D$\times 1$D mesh on which the
recursion in \Cref{lem:update} fails. We then give a concrete game instance
whose induced welfare functions violate the same recursion, establishing that
the failure can arise in a game with three Bob types, two Alice types, and two
actions.

These examples show that the exact discretization in \Cref{lem:update} is a
genuinely binary-type phenomenon. In the next section, we return to the binary
setting and study the implications of the exact discretization there.

\section{Properties of Optimal Ex-post IR Bayesian Conversations}
We next study the properties of the optimal Bayesian conversations. We make two
key observations: (1) The optimal ex-post IR Bayesian conversation is strictly
weaker than the mediator protocol in terms of the ability to improve social
welfare; (2) There are two possible convergence regimes of ex-post IR Bayesian conversations where
both Alice and Bob have two types and Alice has two available actions: in one,
the optimal welfare is achieved within 6 rounds; in the other, it requires
infinitely many rounds for convergence.
\subsection{Gap between Bayesian conversations and mediator protocols}\label{sec:alg-example}

The dynamic program above highlights a structural restriction of Bayesian
conversations: beliefs are refined one coordinate at a time. Mediator protocols
do not face this restriction, since the mediator can condition a public signal
on both players’ reported types. A natural question is whether this additional
flexibility can lead to strictly higher social welfare under ex-post IR
constraints.

In this subsection, we answer this question affirmatively via a game between an
employer and a job candidate.

\begin{theorem} \label{thm:expost_gap}
    There exists a game where the highest social welfare that can be implemented by ex-post IR committed Bayesian conversations is lower than the highest social welfare that can be implemented by ex-post IR mediator protocols.%
\end{theorem}

We provide an example. Consider a two-player game between employers (Alice) and
job candidates (Bob). Suppose there are two types of employers: employers who
prefer candidates with good programming skills and employers who prefer
candidates with good communication skills, so we have $\theta_A \in
    \{\text{Prog, Comm} \}$. The candidates also have two types: good at
programming and good at communication $\theta_B \in \{\text{Prog, Comm} \}$. An
employer can choose to hire or not hire a candidate $r\in\{\text{hire, not
        hire}\}$ and an employer's utility function is shown in
Table~\ref{table:employer}.
\begin{table}[!h]
    \centering
    \caption{The employer's utility function in the hiring problem.}
    \label{table:employer}
    \setlength{\tabcolsep}{3pt}
    \begin{minipage}[t]{0.49\textwidth}
        \centering
        \begin{tabular}{|c|c|c|}
            \hline $u_A(\text{Prog}, \cdot)$ & $\theta_B = \text{Prog}$ & $\theta_B = \text{Comm}$ \\
            \hline hire                      & $10$                     & $-10$                    \\
            \hline not hire                  & $0$                      & $0$                      \\
            \hline
        \end{tabular}
    \end{minipage}\hfill
    \begin{minipage}[t]{0.49\textwidth}
        \centering
        \begin{tabular}{|c|c|c|}
            \hline $u_A(\text{Comm}, \cdot)$ & $\theta_B = \text{Prog}$ & $\theta_B = \text{Comm}$ \\
            \hline hire                      & $-1$                     & $1$                      \\
            \hline not hire                  & $0$                      & $0$                      \\
            \hline
        \end{tabular}
    \end{minipage}
\end{table}
And a candidate simply wants to be hired,
$u_B = 2 \cdot \bm{1}(r = \text{hire})$.
Suppose $\P(\theta_A = \text{Prog}) = 0.5$ and $\P(\theta_B = \text{Prog}) = 0.6$, so without any communication a Prog-type employer always hires and a Comm-type employer never hires. Suppose we want to maximize the expected social welfare $\E[u_A + u_B]$.

We show that the optimal ex-post IR mediator protocol achieves a social welfare
of $22/5$ whereas the optimal ex-post IR Bayesian conversation only achieves
$21/5$.

In the optimal mediator protocol, the mediator sends one of two public signals
$s_1$, indicating ``good communication skill'' and $s_2$, indicating ``good
programming skill''. It reveals Bob's
type perfectly to Prog-type employers, but garbles it for Comm-type employers
by sometimes sending $s_1$ when $\theta_B=\text{Prog}$.
Appendix~C details the protocol and establishes its
welfare of $22/5$.

The same appendix applies our dynamic program to show that the optimal
ex-post IR conversation achieves $21/5$ with a single message from Bob,
unconditioned on Alice's type (\Cref{fig:optimal_bayesian_conversation}). The
intuition is that if Bob were to condition on Alice’s type, some types of Bob
would have regretted the outcome of the conversation. The corresponding
$\IR_0$, $W_0$, and $W_k$ are plotted in
\Cref{fig:final_response,fig:utility_round_1}.

\begin{figure}[h]
    \centering
    \begin{subfigure}[t]{0.4 \textwidth}
        \begin{tikzpicture}
            \begin{axis}[
                    xmin=-0.03,
                    xmax=1.15,
                    ymin=-0.03,
                    ymax=1.15,
                    axis lines = middle,
                    xlabel = {$q_B(\theta_B = \text{Prog})$},
                    ylabel = {$q_A(\theta_A = \text{Prog})$},
                    width= \textwidth,
                    height = 0.9\textwidth,
                    x label style={at={(axis description cs:0.5,-0.15)},anchor=north},
                    y label style={at={(axis description cs:-0.17,.5)},rotate=90,anchor=south},
                    xtick={0, 0.5, 0.6, 1},
                    xticklabels = {$0$, {$0.5\hspace{8pt}$}, {$\hspace{8pt}0.6$}, $1$},
                ]
                \fill[gray, fill opacity=0.5] (axis cs:0.5,0.5) rectangle (axis cs:1,1);
                \fill[gray, fill opacity=0.5] (axis cs:0,0) rectangle (axis cs:0.5,0.5);
                \filldraw[blue] (0.6,0.5) circle (2pt) node[anchor=north west]{$(q_B^0, q_A^0)$};
                \node at (axis cs:0.75,0.75) {$\IR_0$};
                \node at (axis cs:0.25,0.25) {$\IR_0$};
            \end{axis}
        \end{tikzpicture}
        \caption{$\IR_0$}
        \label{fig:final_response}
    \end{subfigure}
    \hspace{0.08\textwidth}
    \begin{subfigure}[t]{0.4\textwidth}
        \begin{tikzpicture}
            \begin{axis}[
                    xmin=-0.03,
                    xmax=1.15,
                    ymin=-0.03,
                    ymax=1.15,
                    axis lines = middle,
                    xlabel = {$q_B(\theta_B = \text{Prog})$},
                    ylabel = {$q_A(\theta_A = \text{Prog})$},
                    width=\textwidth,
                    height = 0.9\textwidth,
                    clip=false,
                    x label style={at={(axis description cs:0.5,-0.15)},anchor=north},
                    y label style={at={(axis description cs:-0.17,.5)},rotate=90,anchor=south},
                    xtick={0, 0.5, 0.6, 1},
                    xticklabels = {$0$, {$0.5\hspace{8pt}$}, {$\hspace{8pt}0.6$}, $1$},
                ]

                \addplot[color = gray, dashed, line width = 0.3mm,  domain = 0:1]{0.5};
                \addplot[color = gray, dashed, line width = 0.3mm,  domain = 0:1]{1};
                \addplot[color = gray, dashed, line width = 0.3mm] coordinates {(0.5,0)(0.5,1)};
                \addplot[color = gray, dashed, line width = 0.3mm] coordinates {(0.6,0)(0.6,1)};
                \addplot[color = gray, dashed, line width = 0.3mm] coordinates {(1,0)(1,1)};
                \filldraw[blue] (0,1) circle (2pt) node[anchor=south west]{$-\infty$};
                \filldraw[blue] (0.5,1) circle (2pt) node[anchor=south east]{$2$};
                \filldraw[blue] (1,1) circle (2pt) node[anchor=south west,xshift=1pt]{$12$};

                \filldraw[blue] (0,0.5) circle (2pt) node[anchor=south west]{$\frac{3}{2}$};
                \filldraw[blue] (0.5,0.5) circle (2pt) node[anchor=south east]{$2$};
                \filldraw[blue] (1,0.5) circle (2pt) node[anchor=south west,xshift=1pt]{$6$};

                \filldraw[blue] (0,0) circle (2pt) node[anchor=south west]{$3$};
                \filldraw[blue] (0.5,0) circle (2pt) node[anchor=south east]{$2$};
                \filldraw[blue] (1,0) circle (2pt) node[anchor=south west,xshift=1pt]{$-\infty$};

                \filldraw[blue] (0.6,0) circle (2pt) node[anchor=south west]{$-\infty$};
                \filldraw[blue] (0.6,0.5) circle (2pt) node[anchor=south west]{$2$};
                \filldraw[blue] (0.6,1) circle (2pt) node[anchor=south west]{$4$};
            \end{axis}
        \end{tikzpicture}
        \caption{$W_0$}
        \label{fig:utility_round_0}
    \end{subfigure}

    \caption{Initialization for the hiring problem. Left: the gray rectangles are the terminal beliefs satisfying both types of Bob's ex-post IR constraints. Right: $W_0$ at the grid points; infeasible points have value $-\infty$. Values at $q_B=0.6$ are shown only to locate the prior and are not needed by \Cref{alg:alg1}.}
\end{figure}
\begin{figure}[h]
    \centering
    \begin{subfigure}[t]{0.4\textwidth}
        \begin{tikzpicture}
            \begin{axis}[
                    xmin=-0.03,
                    xmax=1.15,
                    ymin=-0.03,
                    ymax=1.15,
                    axis lines = middle,
                    xlabel = {$q_B(\theta_B = \text{Prog})$},
                    ylabel = {$q_A(\theta_A = \text{Prog})$},
                    width=\textwidth,
                    height = 0.9\textwidth,
                    clip=false,
                    x label style={at={(axis description cs:0.5,-0.15)},anchor=north},
                    y label style={at={(axis description cs:-0.17,.5)},rotate=90,anchor=south},
                    xtick={0, 0.5, 0.6, 1},
                    xticklabels = {$0$, {$0.5\hspace{8pt}$}, {$\hspace{8pt}0.6$}, $1$},
                ]

                \addplot[color = gray, dashed, line width = 0.3mm,  domain = 0:1]{0.5};
                \addplot[color = gray, dashed, line width = 0.3mm,  domain = 0:1]{1};
                \addplot[color = gray, dashed, line width = 0.3mm] coordinates {(0.5,0)(0.5,1)};
                \addplot[color = gray, dashed, line width = 0.3mm] coordinates {(0.6,0)(0.6,1)};
                \addplot[color = gray, dashed, line width = 0.3mm] coordinates {(1,0)(1,1)};
                \filldraw[blue] (0,1) circle (1.5pt) node[anchor=south west]{$-\infty$};
                \filldraw[blue] (0.5,1) circle (1.5pt) node[anchor=south east]{$2$};
                \filldraw[blue] (1,1) circle (1.5pt) node[anchor=south west,xshift=1pt]{$12$};

                \filldraw[blue] (0,0.5) circle (1.5pt) node[anchor=south west]{$\frac{3}{2}$};
                \filldraw[red] (0.5,0.5) circle (1.5pt) node[anchor=south east]{$\frac{15}{4}$};
                \filldraw[blue] (1,0.5) circle (1.5pt) node[anchor=south west,xshift=1pt]{$6$};

                \filldraw[blue] (0,0) circle (1.5pt) node[anchor=south west]{$3$};
                \filldraw[blue] (0.5,0) circle (1.5pt) node[anchor=south east]{$2$};
                \filldraw[blue] (1,0) circle (1.5pt) node[anchor=south west,xshift=1pt]{$-\infty$};

                \filldraw[blue] (0.6,0) circle (2pt) node[anchor=south west]{$-\infty$};
                \filldraw[red] (0.6,0.5) circle (2pt) node[anchor=south west]{$\frac{21}{5}$};
                \filldraw[blue] (0.6,1) circle (2pt) node[anchor=south west]{$4$};
            \end{axis}
        \end{tikzpicture}
        \caption{$W_1=W_2=\dots$}
        \label{fig:utility_round_1}
    \end{subfigure}
    \hspace{0.08\textwidth}
    \begin{subfigure}[t]{0.4\textwidth}
        \begin{tikzpicture}
            \begin{axis}[
                    xmin=-0.03,
                    xmax=1.15,
                    ymin=-0.03,
                    ymax=1.15,
                    axis lines = middle,
                    xlabel = {$q_B(\theta_B = \text{Prog})$},
                    ylabel = {$q_A(\theta_A = \text{Prog})$},
                    width=\textwidth,
                    height = 0.9\textwidth,
                    x label style={at={(axis description cs:0.5,-0.15)},anchor=north},
                    y label style={at={(axis description cs:-0.17,.5)},rotate=90,anchor=south},
                    xtick={0, 0.5, 0.6, 1},
                    xticklabels = {$0$, {$0.5\hspace{8pt}$}, {$\hspace{8pt}0.6$}, $1$},
                ]

                \filldraw[black] (0.6,0.5) circle (1.5pt);
                \draw [-{Stealth[length=2mm, scale width=1.5]},thick] (0.6,0.5)  -- node[above]{$\frac{2}{5}$} (0,0.5);
                \draw [-{Stealth[length=2mm, scale width=1.5]},thick] (0.6,0.5)  -- node[above]{$\frac{3}{5}$} (1,0.5);
            \end{axis}
        \end{tikzpicture}
        \caption{Optimal belief split}
        \label{fig:optimal_bayesian_conversation}
    \end{subfigure}
    \caption{Solution for the hiring problem. Left: after one horizontal concavification the value at the prior is $21/5$, and subsequent rounds do not improve it. Right: Bob's single message splits $q_B^0=0.6$ into $0$ and $1$ with probabilities $2/5$ and $3/5$, while $q_A$ stays fixed.}
\end{figure}

\subsection{The Convergence Structure of Ex-post IR Bayesian Conversations}
The example above shows that optimal welfare can be attained in finitely many rounds.
We next characterize the convergence behavior of Bayesian conversations in the
binary-type, binary-action setting. 

\begin{theorem}\label{thm:converge_BC}
    Consider a Bayesian conversation in which both Alice and Bob have two types, i.e., $|\Theta_A| = |\Theta_B| = 2$, and Alice has two available actions, i.e., $|R| = 2$. Then the optimal social welfare is either achieved within $6$ rounds, or it converges linearly to the optimal value over infinitely many rounds. Moreover, there exists a polynomial-time algorithm that determines which of the two cases occurs.
\end{theorem}
\begin{proof}[Proof Sketch]

    \noindent\textbf{Reduction to alternating concavification on a small mesh.}\enspace
We first show that, in the binary-type and binary-action setting, the dynamic program in Algorithm~1 can be reduced to alternating concavification on a small mesh. Since each type of Alice can switch actions at most once as $q_B$ varies, we immediately have $|X^*|\le 4$. For $Y^*$, by examining the structure of Bob's IR intervals across the at most three regions induced by Alice's action switches, together with the two boundary columns $q_B=0,1$, we obtain $|Y^*|\le 7$. 

Although $|Y^*|\le 7$, we further show that at most $5$ rows are effectively involved in the alternating concavification process. This mainly relies on a key observation: if, at a given $y$-level, feasible points occur only on a boundary column and there is no feasible interior point at the same level, then removing this boundary-only point does not affect either the values or the number of concavification rounds in the remaining region. A detailed analysis is deferred to Appendix~D. In summary, it suffices to consider alternating concavification on a mesh of size at most \(4\times 5\).

    \par\medskip\noindent\textbf{Number of concavification steps.}\enspace
    As proved in Appendix~D, if either dimension of the mesh has size at most $3$, the alternating concavification process must terminate in finitely many steps. Hence, the smallest mesh that can exhibit asymptotic convergence has size $4\times4$. The $4\times4$ windmill example in \Cref{fig:windmill} provides the key insight for our subsequent analysis of the $4\times5$ case.

    \begin{figure}[h]
        \centering
        \begin{tikzpicture}[scale=.85,>=stealth,
                every node/.style={font=\small}]
            \draw[gray!45,line width=.4pt] (0,0) grid (3,3);

            \draw[black,line width=.6pt] (0,2)--(2,2);
            \draw[black,line width=.6pt] (1,2)--(1,0);
            \draw[black,line width=.6pt] (1,1)--(3,1);
            \draw[black,line width=.6pt] (2,3)--(2,1);

            \foreach \x in {0,1,2,3}
            \foreach \y in {0,1,2,3}
            \fill[gray] (\x,\y) circle (1.2pt);

            \fill (1,2) circle (2pt)
            node[above left] {$P_0=0$};
            \fill (1,1) circle (2pt)
            node[below left] {$P_1=0$};
            \fill (2,1) circle (2pt)
            node[below right] {$P_2=0$};
            \fill (2,2) circle (2pt)
            node[above right] {$P_3=0$};

            \fill (0,2) circle (2pt)
            node[left] {$Q_0=1$};
            \fill (1,0) circle (2pt)
            node[below] {$Q_1=1$};
            \fill (3,1) circle (2pt)
            node[right] {$Q_2=0$};
            \fill (2,3) circle (2pt)
            node[above] {$Q_3=0$};
        \end{tikzpicture}
        \caption{An instance in which a windmill-shaped infinite update cycle
            converges to the maximum. All points other than those labeled have
            value $-\infty$. Under alternating concavification, $P_0$ is first
            updated to $1/2$, followed by $P_1$ to $3/4$, $P_2$ to $3/8$, and
            $P_3$ to $3/16$. One can show that these are the only points ever
            updated, always in the order
            $P_0\to P_1\to P_2\to P_3\to P_0\to\cdots$, so the process never
            terminates.}
        \label{fig:windmill}
    \end{figure}

    For equally spaced coordinates, each update averages the endpoint values of the
    corresponding black segment. If $x_m$ denotes the value of $P_3$ after $4m$
    updates, then
    \[
        x_{m+1}=\frac{x_m+3}{16},\qquad x_0=0,\qquad
        x_m=\frac{1}{5}\bigl(1-16^{-m}\bigr)\longrightarrow\frac{1}{5}.
    \]
    Thus, the optimum is approached only asymptotically, with a linear convergence
    rate.

    The example above also explains why, for the $4\times4$ case, the process
either terminates within $5$ steps or continues indefinitely. Suppose some
point is still updated in step $6$, say $P_3$. Since the boundary values no
longer change after the first two steps, only the central $2\times2$ points
can be updated in steps $3$--$6$. Tracing backward the point responsible
for each update therefore forces $P_0,P_1,P_2,P_3$ to be updated
successively in steps $3,4,5,6$. The update of $P_3$ then triggers a new
update of $P_0$, forming an infinite cycle.

    On a $4\times5$ mesh, the central $2\times3$ subgrid admits more update
    patterns. The classification in Appendix~D shows that any
    increase in round~7 forces a persistent cycle, giving termination within
    six rounds or linear convergence. The regimes are distinguished in polynomial time by
    comparing the prior's round-6 value with its limiting value, computed by
    the linear program in Appendix~D.
    
    \par\medskip\noindent\textbf{Existence of the Two Regimes.}\enspace
The analysis above characterizes the possible convergence behavior of binary-type, binary-action games. The instance
in Theorem~\ref{thm:expost_gap} reaches the optimum in finitely many rounds,
whereas the instance in Appendix~E converges to it only
asymptotically over infinitely many rounds. The latter construction follows the
proof above: it amounts to choosing suitable initial mesh values that induce such a windmill-shaped update cycle.
\end{proof}

\paragraph{Beyond binary actions.}
When both players have binary types and Alice has an arbitrary finite action
set $R$, \Cref{lemma:bilinear} bounds both breakpoint sets $X^*,Y^*$ by
$O(|R|)$. Our algorithm therefore computes optimal conversations by alternating
concavification on an
$O(|R|)\times O(|R|)$ mesh. For binary actions, \Cref{thm:converge_BC}
completely characterizes the number of steps: the optimum is attained within
six rounds or approached asymptotically at a linear rate. 

However, extending this
dichotomy to arbitrary action sets remains challenging. Even the $4\times5$
case requires a delicate classification of update patterns, and larger meshes
allow substantially richer interactions among interior points. We currently
do not know how to characterize these long-run dynamics or whether an analogous
finite-versus-asymptotic dichotomy holds for general finite action sets.

\section{Conclusion and Future Directions}

We introduced Bayesian conversations as a model of interactive
communication between two privately informed agents. We developed an exact
dynamic program for computing optimal ex-post IR conversations with binary
types and established a strict welfare gap between conversations and mediator
protocols. In the binary-type and binary-action setting, we proved that
optimal welfare is either attained within six rounds or approached asymptotically
at a linear rate, so additional rounds of communication can remain valuable
indefinitely.

Several limitations suggest natural directions for future work. 
First, the example in Appendix~B.5 shows that our finite-discretization approach does not directly extend beyond binary types, calling for new algorithmic ideas in larger type spaces. 
Second, our finite-versus-asymptotic dichotomy is currently limited to binary actions; although arbitrary finite action sets still yield polynomial-size meshes, their richer update dynamics leave open whether a comparable convergence characterization and efficient decision procedure exist. 
Finally, extending the framework to settings where both agents take payoff-relevant actions is another natural direction.

\bibliographystyle{plain}
\bibliography{ref,reference}
\appendix
\section{Optimal mediator protocol by linear programming}
\label{app:model_LP}

We show that the optimal ex-post IR mediator protocol can be solved by an
exponential-size LP.

Consider a game in which Alice has $m$ possible types: $\Theta_A =
    \{\theta_{A1}, \dots, \theta_{Am}\}$, and Bob has $n$ possible types: $\Theta_B
    = \{\theta_{B1}, \dots, \theta_{Bn}\}$. The set of actions that Alice can take
is denoted by $R = \{r_1, r_2, \dots, r_k\}$. We specify a signal using the
joint probability distribution of player types and signals, as illustrated in
the table below.
\begin{center}
    \begin{tabular}{|c|c|c|c|c|}
        \hline
        $\P(\theta_A, \theta_B, s)$ & $\theta_{B1}$ & $\theta_{B2}$ & $\dots$  & $\theta_{Bn}$ \\
        \hline
        $\theta_{A1}$                & $p_{11}(s)$   & $p_{12}(s)$   & $\dots$  & $p_{1n}(s)$   \\
        \hline
        $\theta_{A2}$                & $p_{21}(s)$   & $p_{22}(s)$   & $\dots$  & $p_{2n}(s)$   \\
        \hline
        $\vdots$                     & $\vdots$      & $\vdots$      & $\ddots$ & $\vdots$      \\
        \hline
        $\theta_{Am}$                & $p_{m1}(s)$   & $p_{m2}(s)$   & $\dots$  & $p_{mn}(s)$   \\
        \hline
    \end{tabular}
    \captionof{table}{Joint distribution for players' types $\theta_A, \theta_B$ and signal $s$}
    \label{fig:ex-post-lp}
\end{center}

To find the optimal ex-post IR mediator protocol, we describe a signal by the
actions taken by different types of Alice under that signal. For instance, in a
game where $|\Theta_A| = 3$, a signal $s(r_3, r_1, r_2)$ represents that an
Alice of type $\theta_{A1}$ will take action $r_3$ under this signal, an Alice
of type $\theta_{A2}$ will take action $r_1$ under this signal, and so on.

Thus, we have a linear programming algorithm to find the optimal ex-post IR
mediator protocol. In this algorithm, we only need to enumerate all possible
kinds of signals that exhaustively represent all combinations of actions that
different types of Alice might take: $\mathcal{S} = \{s(\mathbf{r}):
    \mathbf{r}\in R^m\}$ We define the objective such that each kind of signal in
the set $\mathcal{S}$ appears at most once in the optimal solution (it may also
not appear, when all the corresponding probabilities $p_{ij}(s)$ are 0). The
objective function maximizes social welfare, while constraints are imposed to
ensure the existence of a mediator that can induce the target solution, guide
different types of Alice to select the expected actions, and satisfy Bob's
ex-post IR condition. The resulting linear programming formulation is as
follows:
\begingroup
\allowdisplaybreaks[4]
\begin{align*}
    \max \quad              & \sum_{\substack{s = s(r_1^*, r_2^*, \dots, r_m^*)\\s\in \mathcal{S}}}
                                \sum_{i = 1}^m\sum_{j = 1}^n
                                p_{ij}(s)u(\theta_{Ai},\theta_{Bj},r_i^*) \\
    \text{subject to} \quad & \sum_{s\in\mathcal{S}}p_{ij}(s) = \P(\theta_{Ai})\cdot \P(\theta_{Bj}),\quad \forall i, \forall j                             \\
                            & \sum_{j = 1}^n p_{ij}(s)u_A(\theta_{Ai},\theta_{Bj},r_i^*) \geq \sum_{j = 1}^n p_{ij}(s)u_A(\theta_{Ai},\theta_{Bj},r'),      \\
                            & \qquad\forall s = s(r_1^*, r_2^*, \dots, r_m^*),\quad s\in \mathcal{S}, \\
                            & \qquad\forall i, \forall r' \neq r_i^*\\
                            & \begin{aligned}
                                \sum_i p_{ij}(s)u_B(\theta_{Ai},\theta_{Bj},r_i^*)
                                &\ge \left(\sum_i p_{ij}(s)\right) \\
                                &\quad\left(
                                \sum_{i'} \P(\theta_{Ai'})u_B(\theta_{Ai'},\theta_{Bj},r_{i'}^0)
                                \right),
                              \end{aligned} \\
                            & \qquad\forall s = s(r_1^*, r_2^*, \dots, r_m^*),\quad s\in \mathcal{S}, \\
                            & \qquad\forall j = 1, 2, \dots, n\\
\end{align*}
\endgroup
where $r^0_i = \arg\max_{r \in R} \E_{\theta_B\sim \P(\theta_B)} [u_A(\theta_{Ai}, \theta_B, r)]$ is the best action for Alice of type $\theta_{Ai}$ without any communication.
\begin{proof}
    The first constraint
    guarantees the existence of a mediator protocol capable of implementing the identified probability distribution.
    The second constraint ensures that the best response for each type of Alice corresponds to the action induced by the signal. The third constraint guarantees that each type of Bob satisfies the ex-post IR condition.

    Therefore, what we need to prove is that the previously assumed set of signals
    $\mathcal{S}$ can be used to find the optimal mediator protocol.

    To prove this, we only need to show that: given two signals $s$ and $s'$, under
    which the same type of Alice chooses the same action, and both satisfy Bob's
    ex-post IR condition, then the summation signal $s_{\text{sum}}$ (defined such
    that $\P(\theta_{Ai}, \theta_{Bj}, s_{\text{sum}}) = p_{ij}(s_{\text{sum}})=
        p_{ij}(s) + p_{ij}(s')$) can also induce the corresponding type of Alice to
    take the same action while maintaining Bob's ex-post IR condition.

    This is easy to prove because the summation of the two signals preserves the required marginal probabilities over type profiles, and hence preserves the first constraint.

    The second and third constraints are equivalent to
    \begin{align*}
        \sum_{j=1}^n p_{ij}(s) & \bigl(u_A(\theta_{Ai},\theta_{Bj},r_i^*)
        -u_A(\theta_{Ai},\theta_{Bj},r')\bigr)\geq0,                      \\
\sum_{i=1}^m p_{ij}(s)&
\left(
u_B(\theta_{Ai},\theta_{Bj},r_i^*)
-
\sum_{i'=1}^m
\P(\theta_{Ai'})
u_B(\theta_{Ai'},\theta_{Bj},r_{i'}^0)
\right)
\ge 0.
    \end{align*}

    If $p_{ij}(s)$ and $p_{ij}(s')$ both satisfy these constraints, then summing
    the corresponding inequalities gives $p_{ij}(s_{\text{sum}})= p_{ij}(s) +
        p_{ij}(s')$ also satisfies these constraints, which completes the proof.

    There are $|R|^{|\Theta_A|}$ signals in total, and each of them has
    $|\Theta_A|\cdot |\Theta_B|$ variables, so the size of this linear program is
    $O(|\Theta_A|\cdot|\Theta_B|\cdot|R|^{|\Theta_A|})$.

\end{proof}

\section{Proofs for Section 3}\label{sec:app_sec3}
\subsection{Proof of \Cref{lem:correctness}}
\begin{proof} This can be seen by induction on $k$. It holds for $k=0$ by definition. For odd $k$, let $W^*_k$ be the welfare of the optimal protocol and $W_k$ the function computed above. Assume that Bob's first signal splits Alice's prior belief $q_B^0$ into a distribution  $Q_B'$. For each $q'_B \sim Q'_B$ what follows is a $k-1$ round protocol starting at $(q'_B, q_A)$ where Bob doesn't regret being at the prior. Hence:
    \begin{align*}
        W^*_k(q_B, q_A) = & \E_{q'_B \sim Q'_B}[W^*_{k-1}(q'_B, q_A)] \\=&\E_{q'_B \sim Q'_B}[W_{k-1}(q'_B, q_A)] \leq  W_{k}(q_B, q_A)
    \end{align*}
    since $W_k$ is the maximum. Since $W_k$ corresponds to a valid protocol, $W_{k}(q_B, q_A) = W^*_k(q_B, q_A)$. The argument for even $k$ is the same, swapping the roles of Alice and Bob.
\end{proof}

\subsection{Proof of \Cref{lemma:bilinear} }\label{app:prf_bilinear}
\begin{proof}
Define
\begin{align*}
X^*=\{0,1\}\cup\{\,q_B\in[0,1]:{}
&\exists\theta_A\in\Theta_A\text{ such that}\\
&r^*(\theta_A,q_B)\text{ differs on the left and right of }q_B\,\}.
\end{align*}
This lists all values of $q_B$ at which Alice's decision changes. On each
interval $(x_i^*,x_{i+1}^*)$, Alice's strategy is constant. We therefore define
\[
X_{\mathrm{int}}
=
\left\{
\frac{x_i+x_{i+1}}{2}
\;\middle|\;
i=1,\dots,m-1
\right\},
\]
choosing one representative from each such interval. Thus,
$X_{\mathrm{int}}$ contains one representative for each open interval on which
Alice's strategy is constant, while the points in $X^*$ account for the
strategies selected at the switching points themselves.

As $q_B$ varies, Alice may induce different IR constraints for Bob. For a fixed
belief $x\in[0,1]$, define Bob's IR interval by
\begin{align*}
I(x)
:=
\Bigl\{
q_A\in[0,1]\;\Bigm|\;&
\mathbb{E}_{\theta_A\sim q_A}
\left[
u_B\bigl(\theta_A,\theta_B,r^*(\theta_A,x)\bigr)
\right]\\
&\ge
\mathbb{E}_{\theta_A\sim q_A^0}
\left[
u_B\bigl(\theta_A,\theta_B,r^0(\theta_A)\bigr)
\right],
\quad
\forall\theta_B\in\operatorname{supp}(x)
\Bigr\}.
\end{align*}
Since each IR constraint is affine in $q_A$, $I(x)$ is an interval, possibly
empty or a singleton.

For any interval $I\subseteq[0,1]$, let
\[
\operatorname{End}(I)
:=
\begin{cases}
\{\min I,\max I\}, & I\neq\varnothing,\\
\varnothing, & I=\varnothing.
\end{cases}
\]
If $I$ is a singleton, the two endpoints coincide and contribute only one
element. We then define
\[
Y^*
:=
\bigcup_{x\in X^*\cup X_{\mathrm{int}}}
\operatorname{End}\bigl(I(x)\bigr)\cup \{q_A^0\}.
\]

Recall that
\[
W^*(q_B,q_A)
=
\mathbb{E}_{\theta_B,\theta_A\sim(q_B,q_A)}
\left[
u_A(\theta_A,\theta_B,r^*(\theta_A,q_B))
+
u_B(\theta_A,\theta_B,r^*(\theta_A,q_B))
\right].
\]
By construction of $Y^*$, no IR boundary crosses the interior of any cell
$(x_i^*,x_{i+1}^*)\times(y_j^*,y_{j+1}^*)$. Moreover, Alice's strategy is
constant on each such cell. Hence each cell is either entirely outside
$\IR_0$, or entirely inside $\IR_0$, in which case $W_0$ is bilinear on that
cell.

For each type of Alice, the identity of her optimal action changes at most
$O(|R|)$ times as $q_B$ varies. Since Alice has two types,
$|X^*|=O(|R|)$, and therefore also $|X_{\mathrm{int}}|=O(|R|)$.
Each $q_B\in X^*\cup X_{\mathrm{int}}$ contributes at most two endpoints to
$Y^*$. Thus
\[
|Y^*|
\le 1+
2|X^*\cup X_{\mathrm{int}}|
=
O(|R|).
\]
\end{proof}

\subsection{Proof of \Cref{lem:update}}\label{app:prf_discrete}
\begin{proof}
    We prove the claim by induction on $k$. The proof for cases $k=1$ and $k=2$ are
    special and then $k \geq 3$ follows a general argument.

    \paragraph{Case $k=1$:} We know by Lemma \ref{lemma:bilinear} that $W^*$ is bilinear. Hence once we fix
    a $\hat y$, $W_0(x,\hat y)$ is piecewise linear on $[0,1]$ with segments
    bounded by $x_i \in X^*$. Alice's induced expected utility is continuous at \(x_i\), whereas Bob's induced expected utility may be discontinuous there. Since Alice
    breaks ties in favor of Bob, the value of $u_B(x_i, \hat y)$ is the maximum of
    its left and right limits, hence: $ W_0(x_i,\hat y) = \max\left(\lim_{x\to
            x_i^-}W_0(x,\hat y),\lim_{x\to x_i^+}W_0(x,\hat y)\right). $

    We need to show that given a point $(\hat x, \hat y)$, if $X'$ is the
    distribution in the solution of the optimization problem \eqref{eq:wel_opt_x},
    then:
    \begin{align*}
        W_{1}(\hat x, \hat y) = & \max \E_{x' \sim X'}[W_{0}(x', \hat y)] \text{ s.t. } \E_{x'\sim X'}[x'] = \hat x         \\
        =                       & \max \sum_{i=1}^m p_i W_{0}(x_i, \hat y) \quad \text{ s.t } \sum_{i=1}^m p_i x_i = \hat x
        \text{ and }  p \in \Delta_m
    \end{align*}

    We prove this by contradiction. Suppose there exists $x^*\in \text{supp}(X')$,
    and $x^*\notin X^*$. Assume $x_i < x^* <x_{i+1}$ and $x^* = q x_i + (1-q)
        x_{i+1}$. We know that $W_0(x,\hat y)$ restricted to $(x_i, x_{i+1})$ is linear
    in $x$, and therefore,
    \begin{align*}
        W_{0}(x^*, \hat y)= & W_{0}(q x_i + (1-q) x_{i+1}, \hat y)                                          \\
        \leq                & q\lim_{x\to x_i^+}W_{0}(x, \hat y)+(1-q)\lim_{x\to x_{i+1}^-}W_{0}(x, \hat y) \\
        \leq                & qW_{0}(x_i, \hat y)+(1-q)W_{0}(x_{i+1}, \hat y).
    \end{align*}

    Hence we can substitute $x^*$ with $x_i$ w.p. $q$ and $x_{i+1}$ w.p. $1-q$.

    \paragraph{Case $k=2$:}
    Fix $\hat x$ and consider the function $y \mapsto W_0(\hat x,y)$. For a fixed
    $\hat x$, Alice's action is uniquely determined, and Bob's ex-post IR condition
    is satisfied within a closed interval. Therefore, we conclude that $ W_0(\hat
        x,y) $ restricted to $[y_j,y_{j+1}]$ is linear. Now, let $Y'$ be the
    distribution in the solution of the optimization problem \eqref{eq:wel_opt_y}.
    We want to show that:
    \begin{align*}
        W_2(\hat x,\hat y) = & \max \E_{y' \sim Y'}[W_{1}(\hat x, y')] \text{ s.t. } \E_{y' \sim Y'}[y'] = \hat y \\=& \max \sum_{j=1}^n p_j W_{1}(\hat x, y_j) \quad \text{ s.t }  \sum_{j=1}^n p_j y_j = \hat y \text{ and } p \in \Delta_n
    \end{align*}
    Suppose there exists $y^*\in \text{supp}(Y')$, and $y^*\notin Y^*$. Assume $y_j < y^* <y_{j+1}$ and $y^* = q y_j + (1-q) y_{j+1}$.  We know that $W_0(\hat x, y)$ restricted to $[y_j, y_{j+1}]$ is linear in $y$, therefore,
    \begin{align*}
        W_{1}(\hat x, y^*) = & \sum_{i=1}^m p_i W_{0}(x_i, y^*)                                            \\
        =                    & q\sum_{i=1}^m p_i W_{0}(x_i, y_j)+(1-q)\sum_{i=1}^m p_i W_{0}(x_i, y_{j+1}) \\
        \leq                 & q W_1(\hat x, y_j)+(1-q) W_1(\hat x, y_{j+1}).
    \end{align*}

    Hence we can substitute $y^*$ with $y_j$ w.p. $q$ and $y_{j+1}$ w.p. $1-q$.

    \paragraph{Case $k\geq 3$:}
    We now consider the case where $k \geq 3$. We take the case where $k$ is an odd
    number as an example. The proof for the even case is entirely similar. We
    assume the base condition for round $k-1$ and $k-2$, then from the property of
    concave closure, we have $W_{k-2}(x, y_j)$ restrict to $[x_i, x_{i+1}]$ is
    linear. We want to prove
    \begin{align*}
        W_{k}(\hat x, \hat y) = & \max \E_{x' \sim X'}[W_{k-1}(x', \hat y)] \text{ s.t. } \E_{x'\sim X'}[x'] = \hat x         \\
        =                       & \max \sum_{i=1}^m p_i W_{k-1}(x_i, \hat y) \quad \text{ s.t } \sum_{i=1}^m p_i x_i = \hat x
        \text{ and } p \in \Delta_m
    \end{align*}

    Suppose there exists $x^*\in \text{supp}(X')$, and $x^*\notin X^*$. Assume $x_i
        < x^* <x_{i+1}$ and $x^* = q x_i + (1-q) x_{i+1}$. Then:

    \begin{align*}
        W_{k-1}(x^*,\hat y) = & \sum_{j=1}^n p_j W_{k-2}(x^*, y_j)                                           \\
        =                     & q\sum_{j=1}^n p_jW_{k-2}(x_i, y_j)+(1-q)\sum_{j=1}^np_jW_{k-2}(x_{i+1}, y_j) \\
        \leq                  & q W_{k-1}(x_i,\hat y)+(1-q)W_{k-1}(x_{i+1},\hat y).
    \end{align*}

    Hence we can substitute $x^*$ with $x_i$ w.p. $q$ and $x_{i+1}$ w.p. $1-q$.

    By the induction hypothesis, assume $y_j < \hat y < y_{j+1}$ and $\hat y = q'
        y_j + (1 - q') y_{j+1}$. Then $W_{k-1}$ is linear in $y$ on both segments
    $\{x_i\} \times [y_j, y_{j+1}]$ and $\{x_{i+1}\} \times [y_j, y_{j+1}]$.
    Consequently,
    \begin{align*}
         & W_{k-1}(x_i,\hat y)
        = q'\,W_{k-1}(x_i,y_j) + (1 - q')\,W_{k-1}(x_i,y_{j+1}),
        \\
         & W_{k-1}(x_{i+1},\hat y)
        = q'\,W_{k-1}(x_{i+1},y_j) + (1 - q')\,W_{k-1}(x_{i+1},y_{j+1}).
    \end{align*}

    Therefore, throughout the recursive step, it suffices to track the values on
    the grid $X^* \times Y^*$ only.
\end{proof}

\subsection{Proof of \Cref{thm:algo_correctness} }\label{app:prf_algcorr}
\begin{proof}
    Using Lemma \ref{lemma:bilinear} we can construct the sets $X^*$ and $Y^*$ in
    polynomial time, where each set has size $O(|R|)$. Using the definition of $W_0$
    we can construct a table that stores the value of $W_0$ on $X^* \times Y^*$.
    Note that $q_A^0 \in Y^*$ (see \Cref{app:prf_bilinear}). If $q_B^0\in
    X^*$, the algorithm in \Cref{lem:update} constructs each $W_k$ from
    $W_{k-1}$ on $X^*\times Y^*$.

    If $q_B^0\notin X^*$, let $x_i<q_B^0<x_{i+1}$ be consecutive points in
    $X^*$ and write $q_B^0=qx_i+(1-q)x_{i+1}$. If $T$ is even (which means
    that the \(T\)-th backward update is in the vertical direction), the
    preceding dynamic-programming loop has computed $W_{T-1}$ only on
    $X^*\times Y^*$. We first obtain the required off-grid values by linearity
    in the $x$-direction:
    \begin{align*}
        W_{T-1}(q_B^0,y_j)
        ={}&qW_{T-1}(x_i,y_j)
        +(1-q)W_{T-1}(x_{i+1},y_j)
        \qquad (y_j\in Y^*).
    \end{align*}
    We can then perform the final vertical update as follows:
    \begin{align*}
        W_{T}(q_B^0,q_A^0) =
         & \max \sum_{j=1}^n p_j W_{T-1}(q_B^0, y_j) \\          & \text{ s.t } \quad \sum_{j=1}^ n p_j y_j = q_A^0
                   \quad \text{and} \quad p \in \Delta_n .
    \end{align*}

    If $T$ is odd, by the computation of $W_T$, we know that $W_T$ is linear in $x$
    on the segment $[x_i, x_{i+1}]\times\{q_A^0\}$, hence:
    \begin{align*}
         & W_{T}(q_B^0,q_A^0) = qW_{T}(x_i,q_A^0) + (1- q)W_{T}(x_{i+1},q_A^0).
    \end{align*}

\end{proof}

\subsection{Counterexamples beyond binary types}
\label{app:beyond-binary}

We first isolate the geometric obstruction in a simple mesh example and then
show that the same obstruction arises from the payoffs of a concrete game.

\begin{example}[A geometric obstruction]
    \label{ex:beyond-binary}
    Consider a two-dimensional belief and a one-dimensional belief,
    \[
        p\in\Delta(\{1,2,3\}),\qquad y\in[0,1].
    \]
    Let $Q$ be a nondegenerate parallelogram and $O$ the intersection of its
    diagonals:
    \[
        Q=\operatorname{conv}\{A,B,C,D\}\subseteq\Delta(\{1,2,3\}),
        \qquad O=AC\cap BD=\frac{A+C}{2}=\frac{B+D}{2}.
    \]
    Define
    \[
        W_0(A,y)=W_0(C,y)=2(1-y),\qquad
        W_0(B,y)=W_0(D,y)=2y,
    \]
    and extend $W_0(\cdot,y)$ affinely on the two
    triangles $ABC$ and $ACD$ for each fixed $y$. Thus $W_0$ is bilinear on
    each cell of the induced subdivision of $Q\times[0,1]$. Set $W_0=-\infty$
    outside $Q\times[0,1]$, and let
    $W_1=\operatorname{cav}_p W_0$,
    $W_2=\operatorname{cav}_y W_1$, and
    $W_3=\operatorname{cav}_p W_2$, where
    $\operatorname{cav}_p$ and $\operatorname{cav}_y$ denote concavification
    in the $p$- and $y$-coordinates, respectively.

    At the point $O$, the two equal-weight decompositions
    $O=\tfrac12A+\tfrac12C=\tfrac12B+\tfrac12D$ imply
    $W_1(O,y)=\max\{2(1-y),2y\}$. Hence $W_1(O,0)=W_1(O,1)=2$ and
    $W_1(O,\tfrac12)=1$. Concavifying in $y$ therefore gives
    $W_2(O,\tfrac12)=2$, and so also $W_3(O,\tfrac12)=2$.

    Now restrict the third update to the original subdivision vertices only,
    and denote the resulting value by $\widehat W_3$. Since each
    $v\in\{A,B,C,D\}$ is an extreme point of $Q$ and $W_0(v,\cdot)$ is affine
    in $y$, the first two concavifications do not change its value at
    $y=\tfrac12$; thus
    \[
        W_2\!\left(A,\tfrac12\right)
        =W_2\!\left(B,\tfrac12\right)
        =W_2\!\left(C,\tfrac12\right)
        =W_2\!\left(D,\tfrac12\right)=1.
    \]
    Therefore every convex combination of the available original-vertex
    values with mean $O$ is still equal to $1$, so
    \[
        \widehat W_3\!\left(O,\tfrac12\right)
        =1<2=W_3\!\left(O,\tfrac12\right).
    \]
    Thus a new interior point $O$, which is not a vertex of the initial
    subdivision, becomes essential after two concavification steps. This is
    exactly the failure of the original-vertex discretization on which
    \Cref{alg:alg1} relies. The same obstruction extends to larger type spaces
    by padding the additional belief coordinates with zeros. The construction
    is illustrated in \Cref{fig:beyond-binary}.
\end{example}

\begin{figure}[!htb]
    \centering
\begingroup
\definecolor{bbedge}{HTML}{5F5F5F}
\definecolor{bbdash}{HTML}{B0B0B0}
\definecolor{bbaccent}{HTML}{7A3E75}
\begin{tikzpicture}[x=1cm,y=1cm,font=\fontsize{9}{11}\selectfont,
    pt/.style={circle,fill=black,inner sep=1.8pt},
    opt/.style={circle,fill=black,inner sep=2.0pt},
    back/.style={draw=bbdash,line width=.5pt,densely dashed},
    edge/.style={draw=bbedge,line width=.6pt},
    oval/.style={fill=white,inner sep=.8pt,font=\fontsize{9}{10.6}\selectfont},
    olabel/.style={text=bbaccent,fill=white,inner sep=.8pt,font=\fontsize{9}{10.6}\selectfont}]

  \def\dx{1.00}
  \def\dy{0.70}
  \def\wx{2.15}
  \def\hz{1.20}

  \foreach \panel/\xshift/\tag/\title in {0/0/{(a)}/{Current finite-vertex update},1/6.6/{(b)}/{Correct full-space update}}{
  \begin{scope}[shift={(\xshift,0)}]
    \coordinate (A0) at (0,0);
    \coordinate (B0) at (\wx,0);
    \coordinate (D0) at (\dx,\dy);
    \coordinate (C0) at ($(B0)+(\dx,\dy)$);
    \coordinate (Ah) at ($(A0)+(0,\hz)$);
    \coordinate (Bh) at ($(B0)+(0,\hz)$);
    \coordinate (Dh) at ($(D0)+(0,\hz)$);
    \coordinate (Ch) at ($(C0)+(0,\hz)$);
    \coordinate (A1) at ($(A0)+(0,2*\hz)$);
    \coordinate (B1) at ($(B0)+(0,2*\hz)$);
    \coordinate (D1) at ($(D0)+(0,2*\hz)$);
    \coordinate (C1) at ($(C0)+(0,2*\hz)$);
    \coordinate (O0) at ($(A0)!0.5!(C0)$);
    \coordinate (Oh) at ($(Ah)!0.5!(Ch)$);
    \coordinate (O1) at ($(A1)!0.5!(C1)$);

    \node[font=\bfseries] at ($(A1)!0.5!(C1)+(0,0.78)$) {\tag\ \title};

    \draw[edge] (A0)--(B0)--(C0)--(D0)--cycle;
    \draw[edge] (Ah)--(Bh)--(Ch)--(Dh)--cycle;
    \draw[edge] (A1)--(B1)--(C1)--(D1)--cycle;
    \draw[edge] (A0)--(A1) (B0)--(B1);
    \draw[back] (D0)--(D1) (C0)--(C1);
    \draw[back] (Ah)--(Bh) (Dh)--(Ch);
    \draw[back] (Ah)--(Dh) (Bh)--(Ch);

    \node[left] at ($(A0)+(-0.22,0.02)$) {$y=0$};
    \node[left] at ($(Ah)+(-0.22,0.00)$) {$y=\tfrac12$};
    \node[left] at ($(A1)+(-0.22,0.00)$) {$y=1$};

    \foreach \P in {A0,B0,C0,D0,Ah,Bh,Ch,Dh,A1,B1,C1,D1}{\node[pt] at (\P) {};}

    \node[below=4pt] at (A0) {$A$};
    \node[below=4pt] at (B0) {$B$};
    \node[below right=3pt and -1pt] at (C0) {$C$};
    \node[below left=3pt and -1pt] at (D0) {$D$};

    \node[oval,above left=1pt and 1pt] at (A0) {$2$};
    \node[oval,above right=1pt and 1pt] at (B0) {$0$};
    \node[oval,above right=1pt and 1pt] at (C0) {$2$};
    \node[oval,above left=1pt and 1pt] at (D0) {$0$};

    \node[oval,above left=1pt and 1pt] at (Ah) {$1$};
    \node[oval,above right=1pt and 1pt] at (Bh) {$1$};
    \node[oval,above right=1pt and 1pt] at (Ch) {$1$};
    \node[oval,above left=1pt and 1pt] at (Dh) {$1$};

    \node[oval,above left=1pt and 1pt] at (A1) {$0$};
    \node[oval,above right=1pt and 1pt] at (B1) {$2$};
    \node[oval,above right=1pt and 1pt] at (C1) {$0$};
    \node[oval,above left=1pt and 1pt] at (D1) {$2$};

    \ifnum\panel=0
      \node[circle,fill=bbaccent,inner sep=0.55pt] at (Oh) {};
      \node[olabel,above=3pt] at (Oh) {$O$};
      \node[oval] at ($(A0)!0.5!(C0)+(0,-1.2)$) {$\widehat W_3(O,\tfrac12)=1$};
    \else
      \draw[back] (A0)--(C0);
      \draw[back] (B1)--(D1);
      \draw[back] (Ah)--(Ch);
      \draw[back] (Bh)--(Dh);
      \draw[edge] (O0)--(Oh)--(O1);
      \node[opt] at (O0) {};
      \node[opt] at (Oh) {};
      \node[opt] at (O1) {};
      \node[olabel,right=3pt] at (O0) {$2$};
      \node[olabel,right=3pt] at (Oh) {$2$};
      \node[olabel,right=3pt] at (O1) {$2$};
      \node[olabel,above=3pt] at (Oh) {$O$};
      \node[oval] at ($(A0)!0.5!(C0)+(0,-1.2)$) {$W_3(O,\tfrac12)=2$};
    \fi
  \end{scope}
  }
\end{tikzpicture}
\endgroup
    \caption{A $2$D$\times 1$D counterexample to the direct extension of
    Algorithm~\ref{alg:alg1}. The vertical direction represents the scalar
    belief $y$, and the numbers written next to the marked points are function
    values. In panel~(a), if one keeps only the original subdivision vertices,
    then at the middle slice $y=\tfrac12$ all available values are equal to
    $1$. In panel~(b), the full-space computation reveals an additional
    relevant point $O$: after the first $p$-concavification, $O$ has value $2$
    at $y=0$ and $y=1$, and the subsequent $y$-concavification yields
    $W_2(O,\tfrac12)=2$. Hence the next $p$-concavification cannot be computed
    correctly from the original vertices alone.}
    \label{fig:beyond-binary}
\end{figure}


\begingroup
\newcommand{\cav}{\operatorname{cav}}
\newcommand{\conv}{\operatorname{conv}}
\newcommand{\V}{\mathcal V}
\newcommand{\half}{\tfrac12}

\begin{example}[A game realizing the failure]
\label{ex:beyond-binary-game}
There is a game with two Alice types, three Bob types, and two actions for which
the third concavification cannot be computed by restricting its support to the
vertices of the initial action-region subdivision. At the independent uniform
prior, the true value is $2/3$, whereas this restriction gives $1/2$.
This is a payoff-derived instance of the obstruction in
\Cref{ex:beyond-binary} and shows that the failure can arise in the model.

\paragraph{1. Game, prior, and individual rationality.}
Let
\[
 \Theta_A=\{a_0,a_1\},\qquad
 \Theta_B=\{b_1,b_2,b_3\},\qquad R=\{0,1\}.
\]
Alice alone chooses the action. Action $0$ gives both players zero for every
type profile. The complete payoff table for action $1$ is
\[
\begin{array}{c|ccc}
 (u_A,u_B)\text{ under }r=1 & b_1&b_2&b_3\\ \hline
 a_0 &(1,0)&(-1,2)&(-1,0)\\
 a_1 &(1,0)&(-1,0)&(-1,2)
\end{array}
\]
Types are independent, with
\[
 \P(a_0)=\P(a_1)=\half,\qquad
 \P(b_1)=\P(b_2)=\P(b_3)=\tfrac13.
\]
Write $p=(p_1,p_2,p_3)$ for the posterior over Bob's type and $y=\P(a_1)$.
For either Alice type, the expected utility of action $1$ is
$p_1-p_2-p_3=2p_1-1$; action $0$ gives zero. Thus
\[
 r^*(a_i,p)=
 \begin{cases}
 0,&p_1<\half,\\
 1,&p_1\ge\half,
 \end{cases}
\]
where at equality we choose action $1$, consistently with Bob-favoring
tie-breaking. If Bob is also indifferent, either choice gives the same welfare.
At the prior, action $0$ is uniquely optimal because $2/3-1=-1/3$.
Consequently, every Bob type has no-communication payoff zero.
Since every entry of $u_B$ under either action is nonnegative, every posterior
satisfies Bob's terminal ex-post IR constraint:
\[
 \sum_i\P(a_i\mid h)\,u_B(a_i,b_j,r^*(a_i,p))
 \ \ge\ 0
 \ =\ \sum_i\P(a_i)\,u_B(a_i,b_j,r^0(a_i)).
\]
In particular, $\IR_0=\Delta(\Theta_B)\times[0,1]$.

\paragraph{2. The induced terminal welfare and its initial vertices.}
For $p_1\ge\half$, the conditional welfare is $1-2p_3$ if Alice has
type $a_0$, and $1-2p_2$ if she has type $a_1$. Therefore
\begin{equation}
 W_0(p,y)=
 \begin{cases}
 0,&p_1<\half,\\
 (1-y)(1-2p_3)+y(1-2p_2),&p_1\ge\half .
 \end{cases}
 \label{eq:explicit-w0}
\end{equation}
This formula follows directly from the payoff table; it is not an independently
specified terminal-value function.

Let $e_1,e_2,e_3$ be the simplex vertices and define
\[
 U=(\half,\half,0),\qquad V=(\half,0,\half),\qquad
 O=(\tfrac13,\tfrac13,\tfrac13).
\]
The only action boundary is $p_1=\half$, the segment $UV$.
It partitions the simplex into the triangle $\conv\{e_1,U,V\}$ and the
quadrilateral $\conv\{U,e_2,e_3,V\}$. Their vertex set is
\[
 \V=\{e_1,e_2,e_3,U,V\},\qquad O\notin\V.
\]
The quadrilateral can be triangulated without adding vertices. There are no
interior $y$-boundaries. Including Alice's prior as in \Cref{alg:alg1},
the initial product grid is $\V\times\{0,\half,1\}$.
As allowed in the original discretization, $W_0$ is bilinear on each open cell
and may have an upward jump at an action boundary.

\paragraph{3. The first concavification: Bob's last message.}
Set $W_1=\cav_p W_0$. For every fixed $y$, the first concavification is
captured exactly by the five original vertices:
\[
\begin{array}{c|ccccc}
 v &e_1&e_2&e_3&U&V\\ \hline
 W_0(v,y)&1&0&0&1-y&y
\end{array}
\]
Indeed, $W_0$ is affine on the upper triangle and is zero below the cutoff;
the nonnegative boundary values dominate any interpolation needed there.
The resulting envelope is
\begin{equation}
 W_1(p,y)=
 \begin{cases}
 \min\{2(1-y)p_1,\ p_1+(1-2y)p_2\},&0\le y\le\half,\\[2pt]
 \min\{2yp_1,\ p_1+(2y-1)p_3\},&\half\le y\le1 .
 \end{cases}
 \label{eq:explicit-w1}
\end{equation}
Here is an explicit verification of the first case. The value at $V$ is at
most the value $1/2$ obtained by splitting $V$ equally into $e_1,e_3$.
If $p_1\le p_2$, the decomposition
\[
 p=2p_1U+(p_2-p_1)e_2+p_3e_3
\]
attains $2(1-y)p_1$. If $p_1\ge p_2$, the decomposition
\[
 p=2p_2U+(p_1-p_2)e_1+p_3e_3
\]
attains $p_1+(1-2y)p_2$. These are convex combinations, and the minimum
of the two displayed affine functions is a concave majorant of $W_0$.
This proves the formula; the other case follows by exchanging $p_2,p_3$
and $y,1-y$.

In particular, writing $F_i(p)=W_1(p,i)$ gives
\begin{equation}
 F_0(p)=\min\{2p_1,1-p_3\},\qquad
 F_1(p)=\min\{2p_1,1-p_2\}.
 \label{eq:explicit-roofs}
\end{equation}
At the prior posterior $O$,
\[
 W_0(O,y)=0,\qquad
 W_1(O,y)=\tfrac23\max\{1-y,y\}.
\]
Thus $W_1(O,0)=W_1(O,1)=2/3$, while $W_1(O,1/2)=1/3$.

\paragraph{4. The second and subsequent concavifications.}
For each fixed $p$, $W_1(p,y)$ is convex in $y$: it is a supremum of
affine functions obtained from feasible $p$-splits. Its concave envelope in
$y$ is therefore precisely the chord joining its endpoint values. Hence
\begin{equation}
 W_2(p,y)=\cav_y W_1(p,y)
 =(1-y)\min\{2p_1,1-p_3\}+y\min\{2p_1,1-p_2\}.
 \label{eq:explicit-w2}
\end{equation}
For fixed $y$ this is concave in $p$, and for fixed $p$ it is affine in $y$.
Thus all subsequent updates leave it unchanged:
\[
 \begin{gathered}
 W_k=W_2\quad(k\ge2),\\[3pt]
 \boxed{W_0(O,\half)=0,\quad W_1(O,\half)=\tfrac13,\quad
 W_2(O,\half)=W_3(O,\half)=\tfrac23.}
 \end{gathered}
\]

\paragraph{5. Why the original-vertex restriction fails.}
Even granting the exact values of $W_2$ at every original vertex, at the
middle slice we have
\[
\begin{array}{c|ccccc}
 v&e_1&e_2&e_3&U&V\\ \hline
 W_2(v,\half)&1&0&0&\tfrac34&\tfrac34
\end{array}
\]
Every entry satisfies $W_2(v,\half)\le\frac32v_1$. Consequently, restricting
the third $p$-concavification to original vertices gives
\[
 \widehat W_3(O,\half)
 :=\max_{\substack{\lambda_v\ge0,\ \sum_{v\in\V}\lambda_v=1\\
                         \sum_{v\in\V}\lambda_vv=O}}
       \sum_{v\in\V}\lambda_v W_2(v,\half)
 \ \le\ \tfrac32 O_1=\tfrac12.
\]
Equality is attained by $O=\frac23U+\frac13e_3$. Therefore
\begin{equation}
 \boxed{\widehat W_3(O,\half)=\tfrac12
       <\tfrac23=W_3(O,\half).}
 \label{eq:explicit-gap}
\end{equation}
This isolates the failure of the third-update support restriction. It does
not claim that the off-grid value $W_2(O,\half)$ is itself impossible to
query directly: the failure is that original vertices cannot represent the
next $p$-concavification exactly.
The geometry of the construction and the alternating updates are illustrated
in \Cref{fig:game-thm32-schematic}.

\paragraph{6. A concrete optimal conversation.}
The value $2/3$ is attained by the following two messages, with Bob speaking last.
\begin{enumerate}
 \setlength{\itemsep}{2pt}
 \item Alice reveals whether she is $a_0$ or $a_1$.
 \item If she is $a_0$, Bob sends one message for $\{b_1,b_2\}$ and another
       for $\{b_3\}$. If she is $a_1$, he sends one message for
       $\{b_1,b_3\}$ and another for $\{b_2\}$.
\end{enumerate}
For $a_0$, Bob's resulting posterior split is
$O=\frac23U+\frac13e_3$; for $a_1$, it is
$O=\frac23V+\frac13e_2$.
On each pooled branch Alice is indifferent and chooses action $1$ in Bob's
favor; welfare is exactly $1$. On the remaining branch she chooses action
$0$, yielding welfare $0$. Hence the expected welfare is $2/3$.
Every terminal branch satisfies ex-post IR, as checked above.
A three-message protocol attains the same value by adding an initial
uninformative Bob message.

\begin{figure}[h]
\definecolor{accent}{RGB}{24,86,133}
\definecolor{warm}{RGB}{156,66,36}
\centering
\resizebox{0.8\linewidth}{!}{%
\begin{tikzpicture}[font=\normalsize, line cap=round, line join=round]
  \begin{scope}[shift={(0,1.45)}]
    \coordinate (E1) at (0,3);
    \coordinate (E2) at (-2,0);
    \coordinate (E3) at (2,0);
    \coordinate (U) at (-1,1.5);
    \coordinate (V) at (1,1.5);
    \coordinate (O) at (0,1);
    \fill[gray!8] (U)--(E2)--(E3)--(V)--cycle;
    \draw[gray!65,line width=.55pt] (E1)--(E2)--(E3)--cycle;
    \draw[line width=.85pt] (U)--(V);
    \draw[accent,dashed,line width=.85pt] (U)--(E3);
    \draw[warm,dashed,line width=.85pt] (V)--(E2);
    \foreach \P in {E1,E2,E3,U,V}{\fill (\P) circle (1.7pt);}
    \fill[white] (O) circle (2.5pt);
    \draw[accent,line width=1pt] (O) circle (2.5pt);
    \node[above] at (E1) {$e_1$};
    \node[below left] at (E2) {$e_2$};
    \node[below right] at (E3) {$e_3$};
    \node[left=3pt] at (U) {$U$};
    \node[right=3pt] at (V) {$V$};
    \node[above=3pt,fill=white,inner sep=1pt] at (O) {$O$};
    \node at (0,2.3) {$r^*=1$};
    \node at (0,.35) {$r^*=0$};
    \node[above=2pt,fill=white,inner sep=1pt] at (0,1.5) {$p_1=\frac12$};
    \node[align=center] at (0,-.85)
      {$\mathcal V=\{e_1,e_2,e_3,U,V\}$\\[2pt]$O\notin\mathcal V$};
  \end{scope}
  \node[font=\normalsize\bfseries] at (0,5.65) {(a) Initial belief subdivision};

  \begin{scope}[shift={(6.5,0)}]
    \foreach \s/\h in {0/0,1/2,2/4}{
      \coordinate (A\s) at (0,{.95+\h});
      \coordinate (B\s) at (-1.85,\h);
      \coordinate (C\s) at (1.85,\h);
      \coordinate (U\s) at (-.925,{.475+\h});
      \coordinate (V\s) at (.925,{.475+\h});
      \coordinate (O\s) at (0,{.316667+\h});
      \draw[gray!55,line width=.5pt] (A\s)--(B\s)--(C\s)--cycle;
      \draw[gray!55,line width=.5pt] (U\s)--(V\s);
      \foreach \P in {A,B,C,U,V}{\fill[gray!75] (\P\s) circle (1.2pt);}
    }
    \foreach \P in {A,B,C}{\draw[gray!35,line width=.45pt] (\P0)--(\P2);}
    \draw[line width=.95pt] (U0)--(C0);
    \draw[line width=.95pt] (V2)--(B2);
    \draw[accent,line width=1pt] (O0)--(O2);
    \foreach \s in {0,1,2}{\fill[accent] (O\s) circle (2pt);}
    \node[left=6pt] at (B0) {$y=0$};
    \node[left=6pt] at (B1) {$y=\frac12$};
    \node[left=6pt] at (B2) {$y=1$};
    \node[below left=2pt] at (U0) {$U:1$};
    \node[below=4pt] at (C0) {$e_3:0$};
    \node[above right=2pt] at (V2) {$V:1$};
    \node[below=4pt] at (B2) {$e_2:0$};
    \node[right=5pt,accent,fill=white,inner sep=1.5pt] at (O0) {$O:\frac23$};
    \node[right=5pt,accent,fill=white,inner sep=1.5pt] at (O2) {$O:\frac23$};
    \node[right=6pt,accent,fill=white,inner sep=2pt,align=left] at (O1)
       {$W_1=\frac13$\\[2pt]$W_2=\frac23$};
    \node[align=center] at (0,-1.2)
      {True third update: $\frac23$\\[2pt]Original vertices only: $\frac12$};
  \end{scope}
  \node[font=\normalsize\bfseries] at (6.5,5.65) {(b) The alternating updates};
\end{tikzpicture}%
}
\caption{The payoff-derived counterexample. Left: $UV$ is the only initial
action boundary. The dashed segments $Ue_3$ and $Ve_2$ are the ridges of
$F_0$ and $F_1$, respectively; they intersect at the new point $O$.
Right: the horizontal slices represent Bob's two-dimensional belief and
the vertical direction represents $y$. The bold black segments are the
optimal Bob splits at $y=0,1$, each giving value $2/3$ at $O$.
Concavification in $y$ raises the middle value from $1/3$ to $2/3$.
The next $p$-update cannot reproduce this value from the original five
vertices.}
\label{fig:game-thm32-schematic}
\end{figure}

\textbf{Geometric insight.}
The two Alice types have the same action threshold, but induce different
welfare-maximizing ways for Bob to pool his types. Their concave envelopes
have different ridges. The common refinement needed after Alice's message
contains $O$, although the original action partition does not.
Thus the failure occurs even with two actions and automatically satisfied IR constraints.
\end{example}
\endgroup

\section{Gap between ex-post IR mediator protocols and Bayesian conversations} \label{app:hiring}

We provide the optimal mediator protocol and Bayesian conversation for the
hiring problem in \Cref{sec:alg-example}.

\paragraph{Optimal mediator protocol.}

By employing the linear program in \Cref{thm:ex-post-lp}, we can get the
optimal ex-post IR mediator protocol $\pi$ that achieves the highest possible
social welfare, which is as follows. Consider a mediator who sends two possible
public signals $s_1=$``good communication skill" and $s_2=$``good programming
skill'' using the following signaling scheme:
\begin{center}
	\begin{tabular}{|c|c|c|}
		\hline $\P(s_1, \cdot)$          & $\theta_B = \text{Prog}$ & $\theta_B = \text{Comm}$ \\
		\hline $\theta_A = \text{Prog}$ & $0$                      & $0.2$                    \\
		\hline $\theta_A = \text{Comm}$ & $0.2$                    & $0.2$                    \\
		\hline
	\end{tabular}

	\begin{tabular}{|c|c|c|}
		\hline $\P(s_2, \cdot)$          & $\theta_B = \text{Prog}$ & $\theta_B = \text{Comm}$ \\
		\hline $\theta_A = \text{Prog}$ & $0.3$                    & $0$                      \\
		\hline $\theta_A = \text{Comm}$ & $0.1$                    & $0$                      \\
		\hline
	\end{tabular}
	\captionof{table}{The joint distribution of the signal and the types.}
	\label{table:signal_app}
\end{center}

An employer's best actions are shown in Table~\ref{table:action_app}.

\begin{centering}
	\centering
	\begin{tabular}{|c|c|c|}
		\hline best action              & $s_1$    & $s_2$    \\
		\hline $\theta_A = \text{Prog}$ & not hire & hire     \\
		\hline $\theta_A = \text{Comm}$ & hire     & not hire \\
		\hline
	\end{tabular}
	\captionof{table}{The employer's best action.}
	\label{table:action_app}
\end{centering}

And the optimal social welfare by ex-post IR mediator protocol is
\begin{align*}
	W = 0.2 \times (0 + 0)+ 0.2 \times (-1 + 2) + 0.2 \times (1 + 2)+ \\0.3 \times (10 + 2) + 0.1 \times (0 + 0) = \frac{22}{5}.
\end{align*}
\paragraph{Optimal Bayesian conversation.} We next employ our algorithm to find the
optimal ex-post IR Bayesian conversation. First of all, the ex-post IR region
$\IR_0$ is plotted as in \Cref{fig:final_response} with
\begin{gather*}
	\IR_0 = [0,0.5]\times [0, 0.5] \cup [0.5, 1]\times [0.5, 1],\\
	X^*=\{0,0.5,1\},\quad Y^*=\{0,0.5,1\}.
\end{gather*}

And the corresponding $W_0(q_B, q_A)$ is plotted in ~\Cref{fig:utility_round_0}
with $W_0(0, 1) = W_0(1, 0)= -\infty$, $W_0(0.5, 0.5)=W_0(0.5,0)=W_0(0.5,1)=2$,
and so on. Then, we perform the recursive updates to find $W_k(\cdot)$ for
$k>0$. The resulting $W_1(q_B, q_A)$ is plotted in ~\Cref{fig:utility_round_1},
and we have $W_2(q_B, q_A) = W_1(q_B, q_A)$, for all $(q_B, q_A)\in X^*\times
	Y^*$. This means that our algorithm has converged after just one step, so the
Bayesian conversation after one round has already achieved the optimal social
welfare, which is equal to $\frac{21}{5}$.

Since \(W_2=W_1\), additional rounds cannot improve upon the one-round protocol. Hence the optimal Bayesian conversation is achieved after a single message from Bob.

\section{Proof of \Cref{thm:converge_BC}}\label{sec:app_newresult}
Prior to the proof, we establish a preliminary lemma.
\begin{lemma}
Consider a Bayesian conversation in which both players have 2 types
and Alice has 2 available actions. Let $X^*$ and $Y^*$ be defined
as in Appendix~\ref{app:prf_bilinear}, with $q_A^0$ included in $Y^*$. Then
$|X^*|\le4$ and $|Y^*|\le7$. Moreover, for every $k$, the value
$W_k(q_B^0,q_A^0)$ can be computed by alternating concavification
on a mesh with at most 4 columns and 5 rows, without changing
the number of iterations.
\end{lemma}

\begin{proof}
Write $p=q_A^0$. We assume that both priors have full support;
degenerate priors leave only one belief coordinate free and are
handled directly.

\smallskip
\noindent\textbf{The bound $|X^*|\le4$.}
Each type of Alice switches actions at most once as $q_B$ varies.
Including $0$ and $1$ therefore gives $|X^*|\le4$.
There are at most three action profiles, including those selected
only at switching points or at the boundary. If three profiles occur,
consecutive profiles differ in the action of exactly one Alice type.

\smallskip
\noindent\textbf{The bound $|Y^*|\le7$.}
For a fixed action profile and a fixed Bob type, the IR surplus is
affine in $q_A$. Unless it is identically zero, it has at most one
zero in $(0,1)$. Every interior endpoint of an IR interval must be
such a zero. At the prior action profile, both surpluses vanish at
$p$, so their only possible isolated interior zero is $p$.
The other two profiles contribute at most four further interior
zeros. Including $0$ and $1$, we obtain
\[
|Y^*|\le 1+4+2=7.
\]
This count also covers the boundary columns: removing one Bob-type
constraint cannot introduce a zero of a new constraint.
If at most two profiles occur, the same argument gives $|Y^*|\le5$.
We henceforth consider the case of three profiles.

\smallskip
\noindent\textbf{At most five effective rows.}
Define
\[
[a,b]
=
\operatorname{conv}\left(
\bigcup_{0<x<1} I(x)
\;\cup\;
\bigl(I(0)\cap I(1)\bigr)
\right).
\]
This interval contains $p$, and its endpoints belong to $Y^*$.
Outside $[a,b]$, no interior point is initially feasible, and no
level is feasible on both boundary columns. These properties persist:
horizontal concavification cannot generate an interior point from
a single boundary point, and vertical concavification cannot leave
the concave hull of the existing interior levels.

Furthermore, on each boundary column, $W_0$ is affine on its IR
interval and equals $-\infty$ elsewhere. It is already concave and
remains unchanged throughout the recursion. Restricting it to
$[a,b]$ preserves these properties. It follows by induction that
restricting the entire recursion to $[0,1]\times[a,b]$ preserves
every $W_k$ on this strip.

If $p=a$ or $p=b$, a mean-preserving split within $[a,b]$ cannot
move the $q_A$-coordinate away from $p$, so only the row $q_A=p$
is needed. Otherwise, $a<p<b$. We prove that
\[
|Y^*\cap[a,b]|\le5.
\]

We use the following observation. For a fixed Bob type, consecutive
profiles change at most one endpoint value of its affine IR surplus.
Thus, if all three constraints have an isolated zero in $(0,1)$,
their feasible directions must agree: changing from a proper
left-IR interval to a proper right-IR interval requires changing
both endpoint signs. Consequently, if a domain endpoint excluded
by this type at the prior becomes feasible under another profile,
at least one of its two nonprior constraints has no isolated
interior zero.

\begin{enumerate}
\item[\textnormal{(a)}]
\textbf{The two prior constraints are nonconstant and have opposite
directions.}
Their intersection is $\{p\}$. Apart from $p$, there are at most
four candidate interior zeros.

If $0$ is retained, it is feasible either at an interior column
or on both boundary columns. In either case, the Bob type whose
prior constraint excludes $0$ must allow $0$ under another profile.
By the observation above, this removes at least one of its two
candidate interior zeros. Retaining $1$ similarly removes a
candidate zero of the other Bob type.

Thus each retained endpoint among $0,1$ is compensated for by
the loss of a different candidate interior zero. Together with
$p$, at most five distinct coordinates remain.

\item[\textnormal{(b)}]
\textbf{The two prior constraints are nonconstant and have the same
direction.}
By symmetry, suppose both are $q_A\ge p$.

If $a=0$, each Bob type must allow $0$ under some profile.
The preceding observation removes at least one candidate interior
zero for each type. Hence there are at most two further interior
coordinates besides $0,1,p$, giving the desired bound.

Suppose $a>0$. Then $0$ is not retained. If any nonprior constraint
has no isolated interior zero, there are already at most five
candidate coordinates: $p$, $1$, and at most three other zeros.

Otherwise, all four nonprior constraints have isolated interior
zeros. The observation implies that all constraints remain
right-IR. Since $a<p$, the lowest retained level is attained either
by an interior IR interval or by the intersection of the two
boundary intervals. In either case,
\[
a=\max\{\ell_1,\ell_2\},
\]
where $\ell_1,\ell_2$ are two of the four nonprior zeros.
If they differ, the smaller lies below $a$ and is discarded;
if they coincide, they contribute only one distinct coordinate.
Thus at most three of the four zeros remain distinct and retained.
Adding $p$ and $1$ again gives at most five coordinates.

\item[\textnormal{(c)}]
\textbf{At least one prior constraint is identically zero.}
For such a Bob type, changing the action of only one Alice type
gives a surplus of the form $c q_A$ or $c(1-q_A)$, which has no
isolated zero in $(0,1)$. Among the other profiles, at most one
differs from the prior profile in both Alice types' actions.
Hence this Bob type contributes at most one additional interior
zero.

If both prior constraints are identically zero, there are at most
two additional interior zeros, so including $0,1,p$ gives at most
five coordinates.

Otherwise, suppose by symmetry that the other prior constraint is
$q_A\ge p$. There are at most three additional interior zeros in
total. If $0$ is not retained, the bound follows immediately.
If $0$ is retained, the nonconstant Bob type must allow it under
another profile, eliminating at least one of that type's candidate
interior zeros. The bound is again five.
\end{enumerate}

Therefore, either only the row $q_A=p$ is needed, or retaining
$Y^*\cap[a,b]$ leaves at most five rows. Since $a,b\in Y^*$,
this restriction introduces no new coordinates.
The exact-discretization argument applies to the retained strip,
and the restriction preserves every iteration. Together with
$|X^*|\le4$, this proves the claimed reduction.
\end{proof}

\begin{lemma}\label{thm:mesh_bound}
    Consider a $4 \times 5$ mesh of points. Suppose we perform alternating concavification along the two coordinate directions. Then the process either terminates in at most $6$ steps, or there exist points whose values strictly increase with each iteration and require infinitely many updates to reach their optimal values. In the latter case, the convergence is linear.
\end{lemma}

The complete proof will be given later in the paper. We first establish several
basic properties concerning alternating concavification on a $4 \times 5$ mesh.
For convenience, we label each point on the $4 \times 5$ grid with a letter as
shown below.

\begin{center}
    \begin{tikzpicture}[scale=1, every node/.style={font=\small}]

        \foreach \y in {0,1,2,3,4}
        \draw[thin] (0,\y) -- (3,\y);
        \foreach \x in {0,1,2,3}
        \draw[thin] (\x,0) -- (\x,4);

        \node[anchor=south east] at (0,4.1) {$a$};
        \node[anchor=south] at (1,4.1) {$b$};
        \node[anchor=south] at (2,4.1) {$c$};
        \node[anchor=south west] at (3,4.1) {$d$};

        \node[anchor=west] at (3.1,3) {$e$};
        \node[anchor=west] at (3.1,2) {$f$};
        \node[anchor=west] at (3.1,1) {$g$};

        \node[anchor=north west] at (3,0) {$h$};
        \node[anchor=north] at (2,0) {$i$};
        \node[anchor=north] at (1,0) {$j$};
        \node[anchor=north east] at (0,0) {$k$};

        \node[anchor=east] at (-0.1,1) {$l$};
        \node[anchor=east] at (-0.1,2) {$m$};
        \node[anchor=east] at (-0.1,3) {$n$};

        \node[anchor=south east] at (1,3) {\textbf{A}};
        \node[anchor=south east] at (1,2) {\textbf{B}};
        \node[anchor=south east] at (1,1) {\textbf{C}};
        \node[anchor=south east] at (2,3) {\textbf{D}};
        \node[anchor=south east] at (2,2) {\textbf{E}};
        \node[anchor=south east] at (2,1) {\textbf{F}};

    \end{tikzpicture}
    \captionof{figure}{Labeling of the $4\times5$ mesh. Uppercase letters denote interior points and lowercase letters denote boundary points.}
\end{center}

\begin{lemma}\label{lem:interior_updates}
    Starting from the 3rd iteration of the concavification process, only the values of the points in the central $2 \times 3$ subgrid are updated.
\end{lemma}

\begin{proof}
    This is because the updates of the boundary rows and columns are independent of the values of the interior points. As a result, the values on the boundary rows and columns remain unchanged after the first $2$ concavification steps.
\end{proof}

\begin{corollary}
    Starting from the 4th concavification step,
    \begin{itemize}
        \item During each odd-numbered (horizontal) update, it is impossible for both $2$
              interior points in any row to be updated simultaneously;
        \item During each even-numbered (vertical) update, it is impossible for all $3$
              interior points in any column to be updated simultaneously.
    \end{itemize}
\end{corollary}

\begin{proof}
    Otherwise, this would imply that at least one of the two endpoints of that row or column was updated in the previous iteration.
\end{proof}

\begin{definition}[Update Arrow]
    For convenience of exposition, we introduce the notion of \emph{update arrows}. For each concavification step, we represent an update by a directed arrow from point $U $ to point $V $, where $V $ is a point whose value is updated in the current step, and $U $ is a point whose value was updated in the previous step and directly caused the update at $V $.
\end{definition}

\begin{center}
    \begin{tikzpicture}[>=stealth, scale=1]

        \foreach \y in {0, 1.5, 3}
        \draw[thin] (-0.3,\y) -- (1.8,\y);
        \foreach \x in {0, 1.5}
        \draw[thin] (\x,-0.3) -- (\x,3.3);

        \node[anchor=south east] at (1.5,3) {$D$};
        \node[anchor=south east] at (1.5,1.5) {$E$};
        \node[anchor=south east] at (1.5,0) {$F$};

        \node[anchor=south west] at (0,0) {$C$};
        \node[anchor=south west] at (0,1.5) {$B$};
        \node[anchor=south west] at (0,3) {$A$};

        \draw[->, blue, thick] (1.5,3) to[out=300, in=60] (1.5,1.5);
        \draw[->, blue, thick] (1.5,3) to[out=300, in=60] (1.5,0);
        \draw[->, blue, thick] (0,0) to[out=120, in=240] (0,1.5);

    \end{tikzpicture}

    \captionof{figure}{Example of update arrows for the 4th (vertical) concavification step. This indicates that in the 3rd round, the values at points C and D were updated, and in the 4th round, the update at D caused updates at E and F, while the update at C caused an update at B.}
\end{center}

\begin{lemma}\label{thm:cycle_form}
    Let a sequence of concavification steps be applied to a $4 \times 5$ mesh. Suppose that for some integer $k \geq 4$ and positive integer $n$, the update diagrams from round $k$ to round $k+n-1$ form a directed cycle. Then:
    \begin{itemize}
        \item There exist points whose values require infinitely many iterations to approach
              their maximal values.
        \item The convergence of these values is linear.
    \end{itemize}
\end{lemma}
\begin{proof}
    We first examine how the values of the points evolve after the cycle is formed. We begin by analyzing the following simple case.
    \begin{center}
        \begin{tikzpicture}[x=1.5cm, y=1.5cm, thick, >=stealth]

            \coordinate (Q0) at (0,2);
            \coordinate (Q1) at (1,0);
            \coordinate (Q2) at (3,1);
            \coordinate (Q3) at (2,3);

            \coordinate (P0) at (1,2);
            \coordinate (P1) at (1,1);
            \coordinate (P2) at (2,1);
            \coordinate (P3) at (2,2);

            \draw (Q0) -- (P3);
            \draw (Q1) -- (P0);
            \draw (Q2) -- (P1);
            \draw (Q3) -- (P2);

            \draw[->, blue, thick] (P0) to[out=240, in=120] (P1);
            \draw[->, blue, thick] (P1) to[out=330, in=210] (P2);
            \draw[->, blue, thick] (P2) to[out=60, in=300] (P3);
            \draw[->, blue, thick] (P3) to[out=150, in=30] (P0);

            \node[left] at (Q0) {$Q_0$};
            \node[left] at (Q1) {$Q_1$};
            \node[right] at (Q2) {$Q_2$};
            \node[right] at (Q3) {$Q_3$};

            \node[above left] at (P0) {$P_0$};
            \node[below left] at (P1) {$P_1$};
            \node[below right] at (P2) {$P_2$};
            \node[above right] at (P3) {$P_3$};

        \end{tikzpicture}
        \captionof{figure}{A four-point update cycle. Each blue arrow points from the previously updated point to the point updated in the current step.}
    \end{center}

    \paragraph{Case $n = 4$.} Firstly we consider the case where the cycle is formed from round $k$ to round
    $k+3$; specifically, for integer $p\in \{0,1,2,3\}$, only the arrow $P_{p \bmod
                4} \rightarrow P_{(p+1) \bmod 4} $ is formed in round $k + p $.

    Consider round $k$. After this round, the values at the points $P_0$, $P_1$,
    and $Q_1$ become linear in the $y$-direction. This linear relationship remains
    unchanged during rounds $k+1$ and $k+2$. In round $k+3$, the value at $P_0$
    increases. Therefore, in round $k+4$, there will still be an arrow from $P_0$
    to $P_1$. By repeating this argument, we see that for any integer $p\geq 0$,
    arrow $P_{p \bmod 4} \rightarrow P_{(p+1) \bmod 4} $ is formed in round $k + p
    $. This indicates that an infinite number of iterations is required.

    In terms of quantitative relations, let $v(M, k) $ denote the value of point $M
    $ at round $k $. Since the value at each point $Q_x $ remains constant
    throughout the process, we use the simplified notation $v(Q_x) $ to denote its
    value. Then we have, for any integer $m$,
    \begin{align*}
        v(P_1, k+4m+1) & = t_0 v(P_0, k+4m) + (1-t_0) v(Q_1),   \\
        v(P_2, k+4m+2) & = t_1 v(P_1, k+4m+1) + (1-t_1) v(Q_2), \\
        v(P_3, k+4m+3) & = t_2 v(P_2, k+4m+2) + (1-t_2) v(Q_3), \\
        v(P_0, k+4m+4) & = t_3 v(P_3, k+4m+3) + (1-t_3) v(Q_0).
    \end{align*}
    Here, for all $x\in \{0,1,2,3\}$, the value of $t_x $ is only determined by the grid spacing, and $t_x\in (0,1)$.

    Multiplying the above four equations, we obtain for any integer $m$:
    \begin{align*}
        v(P_x, k+4m+4) & = T_x v(P_x, k+4m) + c_x.
    \end{align*}
    Here, for all $x\in \{0,1,2,3\}$, the value of $T_x $ and $c_x$ is only determined $t_0, \dots, t_3$ and $c_0 , \dots c_3$, and $T_x\in (0,1)$.

    Since
    \begin{align*}
        v(P_0, \infty) & = T_0 v(P_0, \infty) + c_0 \Rightarrow v(P_0, \infty) = \frac{c_0}{1-T_0}.
    \end{align*}
    we know the error
    \begin{align*}
          & v(P_0, k+4m+4) - v(P_0, \infty)         \\
        = & T_0 v(P_0, k+4m) + c_0 - v(P_0, \infty) \\
        = & T_0 (v(P_0, k+4m) - v(P_0, \infty)).
    \end{align*}
    Therefore, the convergence rate is linear. The same applies to the other points $P_1 $ through $P_3 $.

    \paragraph{Case $n = 2$.} The case $n = 2 $ can be viewed as two simultaneous cycles of the type in the
    $n = 4 $ case, with starting points $P_0 $ and $P_2 $, respectively.

    Based on the analysis of the preceding simple cases, we proceed to consider the
    case of the $4 \times 5$ grid.

    According to \Cref{lem:interior_updates}, starting from the 4th round, the boundary rows and
    columns cannot be updated. Therefore, any vertex involved in a cycle must lie
    in the interior---that is, among the points labeled with uppercase letters.

    We classify the discussion according to the vertices of the directed cycle. Up
    to symmetry, there are 3 distinct cases to consider:

    \paragraph{Case I: Two cycles with the shared edge $B\rightarrow E$.}

    \begin{center}
        \begin{tikzpicture}[>=stealth, scale=1]

            \foreach \y in {0, 1.5, 3}
            \draw[thin] (-0.3,\y) -- (1.8,\y);
            \foreach \x in {0, 1.5}
            \draw[thin] (\x,-0.3) -- (\x,3.3);

            \node[anchor=north east] at (1.5,3) {$D$};
            \node[anchor=south east] at (1.5,1.5) {$E$};
            \node[anchor=south east] at (1.5,0) {$F$};

            \node[anchor=south west] at (0,0) {$C$};
            \node[anchor=south west] at (0,1.5) {$B$};
            \node[anchor=north west] at (0,3) {$A$};

            \draw[->, blue, thick] (0,0) to[out=120, in=240] (0,1.5);
            \draw[->, blue, thick] (0,3) to[out=240, in=120] (0,1.5);
            \draw[->, blue, thick] (1.5,3) to[out=150, in=30] (0,3);
            \draw[->, blue, thick] (0,1.5) to[out=330, in=210] (1.5,1.5);
            \draw[->, blue, thick] (1.5,0) to[out=210, in=330] (0,0);
            \draw[->, blue, thick] (1.5,1.5) to[out=300, in=60] (1.5,0);
            \draw[->, blue, thick] (1.5,1.5) to[out=60, in=300] (1.5,3);

        \end{tikzpicture}
        \captionof{figure}{Update graph for Case I: two directed cycles sharing the edge $B\to E$.}
    \end{center}

    We show that this case is impossible. It suffices to show that the following
    configuration cannot occur: in round $k$ ($k \geq 4$), the arrows $A
        \rightarrow B$ and $C \rightarrow B$ appear, and in round $k+1$, the arrow $B
        \rightarrow E$ appears.

    Suppose that in round $k$ , the arrows $A \rightarrow B$ and $C \rightarrow B$
    appear. Then, at the end of round $k-1$, the values at $n, A, D$ and those at
    $l, C, F$ each satisfy a linear relationship in the $x$-direction. The update
    at round $k$ can be viewed as computing a ``new value'' by taking a linear
    combination of the values from the $n\text{--}e$ row and the $l\text{--}g$ row
    at the end of round $k-1$, and then comparing it to the ``old value'' on the
    $m\text{--}f$ row, keeping the maximum of the two.

    Since the values along the $a\text{--}k$ and $e\text{--}g$ columns are concave
    in the $y$-direction, the ``new value'' at points $m$ and $f$ cannot exceed the
    corresponding ``old value''. Moreover, the $m\text{--}f$ row is concave in the
    $x$-direction at the end of round $k-1$.
    \begin{center}

        \begin{tikzpicture}[scale=0.8, thick]

            \draw[->] (0,0) -- (6.5,0) node[right] {};
            \draw[->] (0,0) -- (0,4) node[above] {\small Value};

            \draw[dashed] (1,0) -- (1,3.5);
            \draw[dashed] (2.5,0) -- (2.5,3.5);
            \draw[dashed] (4,0) -- (4,3.5);
            \draw[dashed] (5.5,0) -- (5.5,3.5);

            \node[below] at (1,0) {\small $m$};
            \node[below] at (2.5,0) {\small $B$};
            \node[below] at (4,0) {\small $E$};
            \node[below] at (5.5,0) {\small $f$};

            \draw[thick] plot[smooth] coordinates {
                    (1,1.1) (2.5,1.6) (4,2.0) (5.5,1.6)
                };

            \draw[thick] (1,0.9) --  (5.5,3.3);

            \node[right] at (5.5,3.3) {\small new};
            \node[right] at (5.5,1.6) {\small old};

        \end{tikzpicture}
        \captionof{figure}{Comparison of the old concave values with the new linear interpolation along the row $m$--$f$.}

    \end{center}

    On the interval from $m$ to $E$, since point $B$ is updated, the new value at
    $B$ is greater than the old one, while the new value at $m$ is no greater than
    the old value. As the new values form a linear segment and the old values are
    concave, it follows that $E$ is also updated to the new value. This implies
    that round $k$ also contains the arrows $D \rightarrow E$ and $F \rightarrow
        E$, which indicates that, prior to round $k$, the values along the
    $n\text{--}e$ and $l\text{--}g$ rows are both linear. As a result, point $f$
    would also be updated. This leads to a contradiction.

    \paragraph{Case II: A single cycle $A\rightarrow B\rightarrow E\rightarrow D\rightarrow A$.}

    \begin{center}
        \begin{minipage}{\textwidth}
            \centering
            \begin{tikzpicture}[>=stealth, scale=1]

                \foreach \y in {0, 1.5, 3}
                \draw[thin] (-0.3,\y) -- (1.8,\y);
                \foreach \x in {0, 1.5}
                \draw[thin] (\x,-0.3) -- (\x,3.3);

                \node[anchor=north east] at (1.5,3) {$D$};
                \node[anchor=south east] at (1.5,1.5) {$E$};
                \node[anchor=south east] at (1.5,0) {$F$};

                \node[anchor=south west] at (0,0) {$C$};
                \node[anchor=south west] at (0,1.5) {$B$};
                \node[anchor=north west] at (0,3) {$A$};

                \draw[->, blue, thick] (0,3) to[out=240, in=120] (0,1.5);
                \draw[->, blue, thick] (1.5,3) to[out=150, in=30] (0,3);
                \draw[->, blue, thick] (0,1.5) to[out=330, in=210] (1.5,1.5);
                \draw[->, blue, thick] (1.5,1.5) to[out=60, in=300] (1.5,3);

            \end{tikzpicture}
            \captionof{figure}{Update graph for Case II: the directed cycle $A\to B\to E\to D\to A$.}
        \end{minipage}
    \end{center}

    We firstly show that, in all rounds after the cycle has formed, the arrow $C
        \rightarrow B $ do not occur. We simply need to show that the route $A
        \rightarrow B \rightarrow E \rightarrow F \rightarrow C$ does not exist.
    Otherwise, when $A\rightarrow B$, the value from $A$ to $C$ is linear with
    respect to $y$-axis. Then when $F\rightarrow C$, there must be $C\rightarrow B$
    in the next step, forming cycle $B \rightarrow E \rightarrow F\rightarrow
        C\rightarrow B$, which contradicts there exists only one cycle $A\rightarrow
        B\rightarrow E\rightarrow D\rightarrow A$.

    We consider the update diagrams over one full cycle (4 rounds) in the case
    where $n = 4$. There are six possible configurations, categorized as follows:

    \begin{center}

        \begin{tikzpicture}[>=stealth, scale=1]

            \foreach \i/\j/\label in {
                    0/1/Type 1,
                    1/1/Type 2,
                    2/1/Type 3,
                    0/0/Type 4,
                    1/0/Type 5,
                    2/0/Type 6
                }
                {

                    \pgfmathsetmacro\xshift{3.8 * \i}
                    \pgfmathsetmacro\yshift{4.8 * \j}

                    \foreach \y in {0, 1.5, 3}
                    \draw[thin] (\xshift - 0.3,\y + \yshift) -- (\xshift + 1.8,\y + \yshift);
                    \foreach \x in {0, 1.5}
                    \draw[thin] (\x + \xshift,\yshift - 0.3) -- (\x + \xshift,\yshift + 3.3);

                    \node[anchor=north east] at (\xshift + 1.5,\yshift + 3) {$D$};
                    \node[anchor=south east] at (\xshift + 1.5,\yshift + 1.5) {$E$};
                    \node[anchor=south east] at (\xshift + 1.5,\yshift + 0) {$F$};

                    \node[anchor=south west] at (\xshift + 0,\yshift + 0) {$C$};
                    \node[anchor=south west] at (\xshift + 0,\yshift + 1.5) {$B$};
                    \node[anchor=north west] at (\xshift + 0,\yshift + 3) {$A$};

                    \draw[->, blue, thick] (\xshift + 0,\yshift + 3) to[out=240, in=120] (\xshift + 0,\yshift + 1.5);
                    \draw[->, blue, thick] (\xshift + 1.5,\yshift + 3) to[out=150, in=30] (\xshift + 0,\yshift + 3);
                    \draw[->, blue, thick] (\xshift + 0,\yshift + 1.5) to[out=330, in=210] (\xshift + 1.5,\yshift + 1.5);
                    \draw[->, blue, thick] (\xshift + 1.5,\yshift + 1.5) to[out=60, in=300] (\xshift + 1.5,\yshift + 3);

                    \ifnum\i=1
                        \ifnum\j=1
                            \draw[->, blue, thick] (\xshift + 0,\yshift + 3) to[out=240, in=120] (\xshift + 0,\yshift + 0);
                        \fi
                    \fi

                    \ifnum\i=2
                        \ifnum\j=1
                            \draw[->, blue, thick] (\xshift + 0,\yshift + 3) to[out=240, in=120] (\xshift + 0,\yshift + 0);
                            \draw[->, blue, thick] (\xshift + 0,\yshift + 0) to[out=330, in=210] (\xshift + 1.5,\yshift + 0);
                        \fi
                    \fi

                    \ifnum\i=0
                        \ifnum\j=0
                            \draw[->, blue, thick] (\xshift + 0,\yshift + 3) to[out=240, in=120] (\xshift + 0,\yshift + 0);
                            \draw[->, blue, thick] (\xshift + 0,\yshift + 0) to[out=330, in=210] (\xshift + 1.5,\yshift + 0);
                            \draw[->, blue, thick] (\xshift + 1.5,\yshift + 1.5) to[out=300, in=60] (\xshift + 1.5,\yshift + 0);
                        \fi
                    \fi

                    \ifnum\i=1
                        \ifnum\j=0
                            \draw[->, blue, thick] (\xshift + 1.5,\yshift + 1.5) to[out=300, in=60] (\xshift + 1.5,\yshift + 0);
                        \fi
                    \fi

                    \ifnum\i=2
                        \ifnum\j=0
                            \draw[->, blue, thick] (\xshift + 0,\yshift + 3) to[out=240, in=120] (\xshift + 0,\yshift + 0);
                            \draw[->, blue, thick] (\xshift + 1.5,\yshift + 1.5) to[out=300, in=60] (\xshift + 1.5,\yshift + 0);
                        \fi
                    \fi

                    \node at (\xshift + 0.75,\yshift - 0.6) {\small \label};
                }

        \end{tikzpicture}
        \captionof{figure}{The six possible update patterns in Case II over one complete four-round cycle.}

    \end{center}

    We analyze the possible types that may arise in the next cycle after completing
    one full cycle.

    \begin{itemize}
        \item \textbf{Type 1:} When the iteration follows Type 1, the value at point $B$ increases with each iteration. As the cycle progresses, it is possible that a convex combination of the values at $A$ and point $j$ eventually exceeds the value at $C$, leading to the appearance of the arrow $A \rightarrow C$. The increase at point $C$ may then immediately trigger $C \rightarrow F$ in the next round. Similar cases may occur for other arrows as well.

              Therefore, after evolving according to Type 1 in one cycle, the next cycle may
              either remain in Type 1 or transition to one of Types $2 \sim 6$. In
              particular, if $A\to C$ and $C\to F$ is triggered, and then the combination of
              the updated values at $F$ and $C$ surpasses that at $E$, the iteration
              transitions to Type 7 in Case III.

        \item \textbf{Type 2 \& 5:}
              By the same reasoning as in Type 1, a cycle of Type 2 or Type 5 may transition in the next round to Type 3, 4, 6, or 7.

              In particular, the transition from Type 5 to Type 3 can be explained as
              follows: the appearance of the arrow $C \rightarrow F$ increases the value at
              point $F$, which may in turn prevent the arrow $E \rightarrow F$ from
              appearing. If this occurs, the cycle transitions to Type 3; otherwise, it
              becomes Type 4.

        \item \textbf{Type 3:} We prove that if a cycle is of Type 3, then the next cycle must also be of Type 3 and cannot transition to any other type. Specifically, we show that neither $E \rightarrow F$ nor $F \rightarrow E$ can occur in the next cycle.

              We focus on the current cycle and let the step in which $A \rightarrow B$
              appears be round $k+1$. Then we have:
              \begin{align*}
                  v(B, k+1) & = t_{AB}v(A, k) + (1-t_{AB})v(j),   \\
                  v(C, k+1) & = t_{AC}v(A, k) + (1-t_{AC})v(j),   \\
                  v(E, k+2) & = t_{BE}v(B, k+1) + (1-t_{BE})v(f), \\
                  v(F, k+2) & = t_{CF}v(C, k+1) + (1-t_{CF})v(g).
              \end{align*}
              Then
              \begin{align*}
                  v(E, k+2) = t_{AB}t_{BE}v(A, k) + c_1, \\
                  v(F, k+2) = t_{AC}t_{CF}v(A, k) + c_2.
              \end{align*}
              Since in the current cycle, there is no $E\to F$ and $F\to E$, we have
              \begin{align*}
                              & v(E, k+2) \geq t_{FE}v(F, k+2) + (1-t_{FE})v(d)     \\
                  \Rightarrow & (t_{AB}t_{BE} - t_{AC}t_{CF}t_{FE})v(A, k) \ge C_1  \\
                              & v(F, k+2) \geq t_{EF}v(E, k+2) + (1-t_{EF})v(i)     \\
                  \Rightarrow & (t_{AC}t_{CF} - t_{AB}t_{BE}t_{EF})v(A, k) \ge C_2,
              \end{align*}
              where $C_{i}$ is a constant determined by the values at outer points, and for any $t$, $t\in(0,1)$.

              Since $t_{AB}>t_{AC}, t_{AB}t_{EF} = t_{AB}t_{BC} > t_{AC}$, and $t_{BE} =
                  t_{CF}$, we know
              \begin{align*}
                  t_1: =t_{AB}t_{BE} - t_{AC}t_{CF}t_{FE}>0 \\
                  t_2: =t_{AC}t_{CF} - t_{AB}t_{BE}t_{EF}>0
              \end{align*}
              and because for any integer $m$, $v(A, k+4m) > v(A, k)$, we have $t_1 v(A, k+4m) \geq C_1$ and $t_2 v(A, k+4m) \geq C_2$ holds for any integer $m$, hence $E\to F$ and $F \to E$ cannot occur in the following cycles, which means the cycle remains Type 3.

        \item \textbf{Type 6:} We prove that if a cycle is of Type 6, then the next cycle must also be of Type 6 and cannot transition to any other type. Specifically, we show that $C \rightarrow F$ cannot occur in the next cycle. ($F \to C$ cannot occur either, as is proven above.)

              We focus on the current cycle and let the step in which $E \rightarrow F$ and
              $E \rightarrow D$ appears be round $k+1$. Then we have:
              \begin{align*}
                  v(F, k+1) & = t_{EF}v(E, k) + (1-t_{EF})v(i),   \\
                  v(D, k+1) & = t_{ED}v(E, k) + (1-t_{ED})v(c),   \\
                  v(A, k+2) & = t_{DA}v(D, k+1) + (1-t_{DA})v(n), \\
                  v(C, k+3) & = t_{AC}v(A, k+2) + (1-t_{AC})v(j).
              \end{align*}
              Then
              \begin{align*}
                  v(F, k+1) & = t_{EF}v(E, k) + c_1             \\
                  v(C, k+3) & = t_{ED}t_{DA}t_{AC}v(E, k) + c_2 \\
              \end{align*}
              Since in the current cycle, there is no $C\to F$, we have
              \begin{align*}
                              & v(F, k+1) \geq t_{CF}v(C, k+3) + (1-t_{CF})v(g)   \\
                  \Rightarrow & (t_{EF} - t_{ED}t_{DA}t_{AC}t_{CF})v(E, k) \ge C.
              \end{align*},
              where $C$ is a constant determined by the values at outer points, and for any $t$, $t\in(0,1)$.

              Since $t_{EF}>t_{AC}$, we know
              \begin{align*}
                  t: =t_{EF} - t_{ED}t_{DA}t_{AC}t_{CF}>0
              \end{align*}
              and because for any integer $m$, $v(E, k+4m) > v(E, k)$, we have $t v(E, k+4m) \geq C$ holds for any integer $m$, hence $C\to F$ cannot occur in the following cycles, which means the cycle remains Type 6.
        \item \textbf{Type 4:} For Type 4, the alternating appearance of $C \rightarrow F$ and $E \rightarrow F$ may result in one preventing the other from forming. As a result, Type 4 may evolve into Type 3 or Type 6.

    \end{itemize}

    We now turn to the case of $n = 2$. In the case of $n = 4$, the process can
    essentially be viewed as the value at point $A$ being updated once every four
    rounds. Based on the amount by which $A$ increases, one can determine the
    subsequent evolution. The case of $n = 2$ can be interpreted as two four-round
    cycles running in parallel, each occupying half of the full loop.

    This means that the value at $A$ is updated once every two rounds, which in
    turn drives the evolution. For each individual four-round cycle, the possible
    update diagrams remain the same as the cases listed above, and the
    directionality of the evolution remains unchanged, since the value at point $A$
    always increases monotonically.

    \paragraph{Case III: A single cycle $A\rightarrow C\rightarrow F\rightarrow D\rightarrow A$.}

    \begin{figure}[htb]
        \centering
        \begin{tikzpicture}[>=stealth, scale=1]

            \foreach \y in {0, 1.5, 3}
            \draw[thin] (-0.3,\y) -- (1.8,\y);
            \foreach \x in {0, 1.5}
            \draw[thin] (\x,-0.3) -- (\x,3.3);

            \node[anchor=north east] at (1.5,3) {$D$};
            \node[anchor=south east] at (1.5,1.5) {$E$};
            \node[anchor=south east] at (1.5,0) {$F$};

            \node[anchor=south west] at (0,0) {$C$};
            \node[anchor=south west] at (0,1.5) {$B$};
            \node[anchor=north west] at (0,3) {$A$};

            \draw[->, blue, thick] (0,3) to[out=255, in=105] (0,0);
            \draw[->, blue, thick] (1.5,3) to[out=150, in=30] (0,3);
            \draw[->, blue, thick] (0,0) to[out=330, in=210] (1.5,0);
            \draw[->, blue, thick] (1.5,0) to[out=75, in=285] (1.5,3);

            \node at (0.75,-0.6) {\small Type 7};

        \end{tikzpicture}
        \caption{Type 7 update pattern.}
    \end{figure}

    If no arrow appears along the edge $BE$ during the update process, then the
    evolution proceeds exactly as in the initial simple case.

    If an update does occur along $BE$, without loss of generality, suppose it is
    $B \rightarrow E$. Since $E$ is updated, the arrow $F \rightarrow D$ in the
    next round may be blocked as a result.

    \begin{itemize}
        \item If the arrow is not blocked, the cycle proceeds as before.
        \item If it is blocked, then the next round of updates can lead to two possibilities:
              \begin{itemize}
                  \item If only $E \rightarrow D$ appears, the system transitions to Type 3;
                  \item If both $E \rightarrow D$ and $E \rightarrow F$ appear, it transitions to Type
                        4.
              \end{itemize}

    \end{itemize}

    In conclusion, we can summarize the evolutionary relationships of the 7 types
    in \Cref{fig:type-transition}.

    \begin{figure}[H]
        \centering
        \begin{tikzpicture}[>=latex, node distance=2cm, thick]

            \node (1) at (0,0) {Type 1};
            \node (2) at (1.7,1) {Type 2};
            \node (5) at (1.7,-1) {Type 5};
            \node (7) at (3.8,1) {Type 7};
            \node (4) at (4.8,-1) {Type 4};
            \node (3) at (6.7,1) {Type 3};
            \node (6) at (6.7,-1) {Type 6};

            \draw[->] (1) -- (2);
            \draw[->] (1) -- (5);
            \draw[->] (2) -- (7);
            \draw[->] (5) -- (7);
            \draw[->] (2) -- (4);
            \draw[->] (5) -- (4);
            \draw[->] (7) -- (3);
            \draw[->] (7) -- (4);
            \draw[->] (4) -- (6);
            \draw[->] (4) -- (3);

        \end{tikzpicture}
        \caption{
            Transition diagram of types. An arrow $i \rightarrow j$ indicates that Type $i$ can evolve into Type $j$. During the evolution process, a type may transition to any other type reachable via a directed path. (Except that Type 7 cannot directly evolve into Type 6.)
        }
        \label{fig:type-transition}
    \end{figure}
\end{proof}

With the above groundwork in place, we are now ready to prove
\Cref{thm:mesh_bound}.

\begin{proof}[Proof of \Cref{thm:mesh_bound}]
    It suffices to show that if some point is still updated in round 7, then the update diagram must form a cycle, and hence the process cannot terminate in finite steps.

    By symmetry, it suffices to consider the cases in which point $A $ or point $B
    $ is updated in round 7.

    \paragraph{Case 1: Point $A $ is updated in round 7.}
    By \Cref{lem:interior_updates}, we know that starting from round 4, if a point is updated, it
    must be caused by an update to some point in the central $2 \times 3 $ region
    in the previous round. In other words, the updated point in the current round
    must have an incoming arrow from a point in the central $2 \times 3 $ region.

    Taking into account that odd-numbered rounds apply horizontal updates and
    even-numbered rounds apply vertical updates, it follows that in round 7 there
    must be an arrow $D \rightarrow A $.

    \begin{enumerate}[label=(\arabic*)]
        \item $E\rightarrow D$ in round 6.

              This implies that $B \rightarrow E$ is in round 5.
              \begin{enumerate}[label=\Alph*.]
                  \item $A\rightarrow B$ in round 4.

                        Then cycle $A\rightarrow B\rightarrow E\rightarrow D \rightarrow A$ is formed.
                  \item $C\rightarrow B$ in round 4.

                        From the proof of \Cref{thm:cycle_form}, we know that such a path $C \rightarrow B
                            \rightarrow E \rightarrow D \rightarrow A$ directly leads to the formation of
                        the arrow $A \rightarrow B$ in round 8. As a result, the cycle $B \rightarrow E
                            \rightarrow D \rightarrow A \rightarrow B$ is formed.

              \end{enumerate}
        \item $F\rightarrow D$ in round 6.

              This implies that $F\rightarrow E$ in round 6, and $C \rightarrow F$ in round
              5.
              \begin{enumerate}[label=\Alph*.]
                  \item $A\rightarrow C$ in round 4.

                        Then cycle $A\rightarrow C\rightarrow F\rightarrow D \rightarrow A$ is formed.
                  \item $B\rightarrow C$ in round 4.

                        This indicates $B$ is updated in round 3, hence the value of $m$, $B$, $E$ is
                        linear in $x$-direction after round 3.
                        \begin{enumerate}[label=\alph*.]
                            \item $B$ is not updated in round 5.

                                  Since in round 6 $E$ is updated, then $E \rightarrow B$ must be in round 7,
                                  therefore cycle $B\rightarrow C\rightarrow F \rightarrow E \rightarrow B$ is
                                  formed.
                            \item $B$ is updated in round 5, i.e. $E \rightarrow B$ in round 5.

                                  Since $B \rightarrow C$ in round 4, $B$ is updated in round 5, $B \rightarrow
                                      C$ must be in round 6. Because $C\rightarrow F$ is in round 5, $F\rightarrow E$
                                  is in round 6, hence cycle $B\rightarrow C\rightarrow F \rightarrow E
                                      \rightarrow B$ is formed.

                        \end{enumerate}
              \end{enumerate}
    \end{enumerate}
    \paragraph{Case 2: Point $B $ is updated in round 7.} This indicates $E\rightarrow
        B$ in round 7. From the proof of \Cref{thm:cycle_form}, we know that $D\rightarrow E$ and
    $F\rightarrow E$ cannot exist simultaneously in round 6. Without loss of generality we assume
    $D\rightarrow E$ in round 6. Hence $A\rightarrow D$ is in round 5.
    \begin{enumerate}[label=(\arabic*)]
        \item $B\rightarrow A$ in round 4.

              Then cycle $B\rightarrow A\rightarrow D\rightarrow E \rightarrow B$ is formed.
        \item $C\rightarrow A$ in round 4.

              The values at $B$, $A$, $b$ is linear in $y$-direction after round 4.
              \begin{enumerate}[label=\Alph*.]
                  \item $A$ is not updated in round 6.

                        Since in round 7, $B$ is updated, $B\rightarrow A$ must be in round 8.
                        Therefore cycle $A\rightarrow D \rightarrow E \rightarrow B \rightarrow A$ is
                        formed.
                  \item $A$ is updated in round 6.
                        \begin{enumerate}[label=\alph*.]
                            \item $B\rightarrow A$ in round 6.

                                  In this case $E\rightarrow B$ must be in round 5. Therefore a cycle
                                  $A\rightarrow D \rightarrow E \rightarrow B\rightarrow A$ is formed.
                            \item $C\rightarrow A$ in round 6.

                                  In this case $F\rightarrow C$ must be in round 5, which indicates the value of
                                  $E$, $F$, $i$ is linear in $y$-direction after round 5. Since $D \rightarrow E$
                                  is in round 6, $D\rightarrow F$ must be in round 6, too. Therefore a cycle
                                  $A\rightarrow D \rightarrow F \rightarrow C \rightarrow A$ is formed.
                        \end{enumerate}
              \end{enumerate}
    \end{enumerate}

    In conclusion, in both cases, a cycle is formed, and hence the process cannot
    terminate in finitely many steps, which completes the proof.
\end{proof}

\begin{lemma}\label{lem:stable_type_algorithm}
    For a given $4\times 5$ mesh, if the update diagrams form a cycle during the alternating concavification process, then there exists a polynomial-time algorithm that determines the type to which the update diagram stabilizes.
\end{lemma}

\begin{proof}
    Given the current type, one can compute a closed-form expression for the value at each point as a function of the iteration number. Once these expressions are obtained, they can be used to check whether the conditions for a type transition are satisfied.

    \begin{itemize}
        \item If there doesn't exist an iteration that satisfies the transition condition,
              then the current type is the final stable type.
        \item If there exists an iteration in which a transition occurs (i.e., a new arrow
              inconsistent with the current type appears), then a full cycle is simulated
              starting from that round to determine the resulting type. The process then
              restarts from the beginning with this new type.
    \end{itemize}

    Since the total number of type transitions is bounded (as shown in
    Figure~\ref{fig:type-transition}), the number of iterations is finite, and the
    final stable type can be determined in polynomial time.

\end{proof}

Analogous to the $4 \times 5$ mesh, the $4 \times 4$ mesh exhibits a similar
property as follows:

\begin{lemma}
    Consider a $4 \times 4$ mesh of points. Suppose we perform alternating concavification along the two coordinate directions. Then the process either terminates in at most $5$ steps, or there exist points whose values strictly increase with each iteration and require infinitely many updates to reach their optimal values. In the latter case, the convergence is linear.
\end{lemma}
\begin{proof}
    Similar to the $4 \times 5$ case, the boundary points stop updating after the 2nd round. Therefore, starting from the 3rd round, it suffices to consider only the interior $2 \times 2$ points.

    \begin{center}
        \begin{tikzpicture}[scale=1, every node/.style={font=\small}]

            \foreach \y in {0,1,2,3}
            \draw[thin] (0,\y) -- (3,\y);
            \foreach \x in {0,1,2,3}
            \draw[thin] (\x,0) -- (\x,3);

            \node[anchor=south east] at (0,3.1) {$a$};
            \node[anchor=south] at (1,3.1) {$b$};
            \node[anchor=south] at (2,3.1) {$c$};
            \node[anchor=south west] at (3,3.1) {$d$};

            \node[anchor=west] at (3.1,2) {$e$};
            \node[anchor=west] at (3.1,1) {$f$};

            \node[anchor=north west] at (3,0) {$g$};
            \node[anchor=north] at (2,0) {$h$};
            \node[anchor=north] at (1,0) {$i$};
            \node[anchor=north east] at (0,0) {$j$};

            \node[anchor=east] at (-0.1,1) {$k$};
            \node[anchor=east] at (-0.1,2) {$l$};

            \node[anchor=south east] at (1,2) {\textbf{A}};
            \node[anchor=south east] at (1,1) {\textbf{B}};
            \node[anchor=south east] at (2,2) {\textbf{D}};
            \node[anchor=south east] at (2,1) {\textbf{C}};

        \end{tikzpicture}
        \captionof{figure}{Labeling of the $4\times4$ mesh. Uppercase letters denote interior points and lowercase letters denote boundary points.}
    \end{center}

    We show that if any point is updated in the 6th round, then a cycle must have
    formed---identical to the simple case analyzed in the proof of
    \Cref{thm:cycle_form}---which implies that the iteration cannot terminate in
    finitely many rounds.

    Since points $A$, $B$, $C$, and $D$ are symmetric, we may assume without loss
    of generality that point $D$ is updated in round 6. That is, the update $C
        \rightarrow D$ occurs in round 6, which implies $B \rightarrow C$ in round 5
    and $A \rightarrow B$ in round 4.

    This means that prior to round 3, the values at points $l$, $A$, and $D$ must
    lie on a line in the $x$-direction.

    \begin{itemize}
        \item If point $A$ is not updated in round 5, then the values at $l$, $A$, and $D$
              remain linear in the $x$-direction before round 6. Since $D$ is updated in
              round 6, it follows that $D \rightarrow A$ occurs in round 7. Therefore, a
              cycle $A \rightarrow B \rightarrow C \rightarrow D \rightarrow A$ is formed
              over rounds 4--7.

        \item If point $A$ is updated in round 5, i.e., $D \rightarrow A$ occurs, then $C
                  \rightarrow D$ must have occurred in round 4. Thus, a cycle $A \rightarrow B
                  \rightarrow C \rightarrow D \rightarrow A$ is formed over rounds 4 and 5.
    \end{itemize}
\end{proof}

For finite meshes with at most 3 rows or columns, we show that no instance
requires infinitely many rounds to converge.

\begin{lemma}
    Consider a finite mesh of points with at most $3$ columns or rows. Suppose we perform alternating concavification along the two coordinate directions. Then the process terminates in at most $4$ steps.
\end{lemma}
\begin{proof}
    As argued previously, the boundary points stop changing after 2 rounds of updates.

    Therefore, when the mesh has only 1 or 2 columns (or rows), the claim is
    immediate.

    If the mesh has 3 columns (or rows), then at round 5, no point in the middle
    column can be updated. Otherwise, it would imply that some point in the
    boundary columns (rows) was updated in round 4, which is impossible.

\end{proof}

Finally, we use all the foundational results above to prove the main theorem
stated at the beginning of this section.

\begin{lemma}
    \label{lem:optimum_poly_time}
For a finite mesh with rational coordinates and rational initial
values, the limiting values of the alternating-concavification
process can be computed in polynomial time.
\end{lemma}

\begin{proof}
Let $\mathcal{G}$ be the set of mesh points and let
\[
S_0=\{x\in\mathcal{G}:W_0(x)>-\infty\}.
\]
Starting from $S_0$, repeatedly add every mesh point lying
between two current points on the same row or column, and let
$S$ be the resulting set. A point outside $S$ can never obtain
a finite value under alternating concavification, so it suffices
to consider the process on $S$.

Introduce one variable $z_x$ for each $x\in S$ and solve
\[
\begin{aligned}
\min \quad&
\sum_{x\in S}z_x\\
\text{s.t.}\quad&
z_x\ge W_0(x),
&&x\in S_0,\\
&
z_b\ge \lambda z_a+(1-\lambda)z_c,
&&b=\lambda a+(1-\lambda)c,
\end{aligned}
\]
where the last constraint is imposed whenever $a,b,c\in S$
lie on the same row or on the same column.

We claim that the optimal solution is exactly
\[
W_\infty(x):=\lim_{k\to\infty}W_k(x).
\]
Let $Z$ be any feasible solution. Since $Z\ge W_0$ and $Z$
is concave on every row and column, induction on $k$ gives
\[
W_k(x)\le Z(x)
\qquad
\text{for every }x\in S.
\]
Indeed, if $Z\ge W_k$, then $Z$ is a concave majorant of
$W_k$ in the direction of the next update, whereas $W_{k+1}$
is the smallest such majorant. Hence
\[
W_\infty(x)\le Z(x).
\]

Conversely, $W_k(x)$ is nondecreasing in $k$ and is bounded
above by any sufficiently large constant function on $S$.
Thus $W_k$ converges pointwise to a finite limit $W_\infty$.
Let $\text{cav}_H$ and $\text{cav}_V$ denote horizontal and
vertical concavification. On a finite mesh these operators are
continuous, and therefore
\[
\text{cav}_H(W_\infty)=W_\infty,
\qquad
\text{cav}_V(W_\infty)=W_\infty.
\]
Thus $W_\infty$ is concave on every row and column and dominates
$W_0$, so it is itself feasible.

Therefore $W_\infty$ is the pointwise smallest feasible
function, and hence it is the optimal solution of the linear
program. The linear program has polynomially many variables and
constraints with rational coefficients of polynomial encoding
length, so it can be solved in polynomial time.
\end{proof}

\begin{proof}[Proof of \Cref{thm:converge_BC}]
    According to the algorithm in \Cref{thm:algo_correctness}, we only need to track the value changes of a single point in the mesh as alternating concavification is applied. By prior analysis, $q_A^0$ must lie in $Y^*$, meaning the point $P(q_B^0, q_A^0)$ representing the prior is located on one of the rows of the constructed mesh.

    If $P$ lies on the boundary of the mesh, then the value converges to the
    maximum within two rounds.

    If $P$ is not on the boundary:
    \begin{itemize}
        \item If the mesh has at most $3$ rows or columns, the process converges within $4$
              rounds.

        \item If the mesh is $4 \times 4$:
              \begin{itemize}
                  \item If alternating concavification terminates within $5$ rounds, then the social
                        welfare also achieves its maximum.
                  \item If it does not terminate in finitely many rounds, then the value along every
                        edge adjacent to points $A$--$D$ grows indefinitely, and the value at point $P$
                        converges linearly to the maximum.
              \end{itemize}

        \item If the mesh is $4 \times 5$:
              \begin{itemize}
                  \item If alternating concavification terminates within $6$ rounds, then the social
                        welfare achieves its maximum.
                  \item If it does not terminate in finitely many rounds, we apply the algorithm from
                        \Cref{lem:stable_type_algorithm} to determine the eventual stable update type.
                        \begin{itemize}
                            \item If the stable type is Type 3, 4, 6, or 7, then since all six interior points
                                  ($A$--$F$) are updated during every full cycle, all edges connected to them are updated
                                  as well. Since $P$ must lie on one of these edges, the social welfare converges
                                  linearly to the maximum.
                            \item If the stable type is Type 1, 2, or 5, then certain edges never get updated:
                                  \begin{itemize}
                                      \item Type 1: edges $lC$, $CF$, and $Fg$,
                                      \item Type 2: edge $Fg$,
                                      \item Type 5: edge $lC$.
                                  \end{itemize}
                                  These edges remain unchanged even after cycles form.
                                  \begin{itemize}
                                      \item If $P$ lies on one of these non-updating edges, then the social welfare reaches
                                            its maximum within $6$ rounds.
                                      \item Otherwise, the social welfare converges linearly to the maximum.
                                  \end{itemize}
                        \end{itemize}
              \end{itemize}
    \end{itemize}
    In summary, in all of the above cases, the value at point $P$ either stops updating within 6 rounds or converges linearly to its maximum value over infinitely many rounds. 

    For the algorithmic part, compute $W_6$ using Algorithm~1 and
compute the limiting value $W_\infty$ using \Cref{lem:optimum_poly_time} above.
The value at the prior is evaluated using the same interpolation
rule as in Algorithm~\ref{alg:alg1}.

If
\[
W_6(q_B^0,q_A^0)=W_\infty(q_B^0,q_A^0),
\]
the optimum is attained within six rounds. Otherwise, by the
structural analysis above, the process is in the asymptotic
regime. Thus the two cases can be distinguished in polynomial
time.

\end{proof}

\section{Existence of Infinite-Round Convergence}
\label{sec:proof_inf_round}

\begin{theorem} \label{thm:infinite_round}
    There exists a game where the highest social welfare achievable through ex-post IR committed Bayesian conversations requires infinitely many rounds to converge to the optimum.
\end{theorem}
\begin{proof}
The proof of Theorem \ref{thm:infinite_round} involves again applying the
methodology in \Cref{thm:algo_correctness} to the following game and then
computing an explicit expression for $W_k(q_B^0, q_A^0)$.

In the two-player game Alice has types $\Theta_A = \{\theta_{A0},
    \theta_{A1}\}$ and Bob has types $\Theta_B = \{\theta_{B0}, \theta_{B1}\}$.
Alice can take two actions $r \in \{r_0, r_1\}$. Suppose $q_B^0 = \P(\theta_B =
    \theta_{B0}) = 0.4$, and $q_A^0=\P(\theta_A = \theta_{A0}) = 0.6$. The
utilities of the two players are given in the following table:

\begin{table}[!h]
    \centering
    \caption{Alice and Bob's utility function.}
    \label{table:inf_game_instance}
    \begin{minipage}{0.48\textwidth}
        \centering
        \begin{tabular}{|c|c|c|}
            \hline $u_A(\theta_{A0}, \cdot, \cdot)$ & $\theta_B = \theta_{B0}$ & $\theta_B = \theta_{B1}$ \\
            \hline $r_0$                            & $7$                      & $5$                      \\
            \hline $r_1$                            & $5$                      & $7$                      \\
            \hline
        \end{tabular}
    \end{minipage}
    \hspace{1em}
    \begin{minipage}{0.48\textwidth}
        \centering
        \begin{tabular}{|c|c|c|}
            \hline $u_A(\theta_{A1}, \cdot, \cdot)$ & $\theta_B = \theta_{B0}$ & $\theta_B = \theta_{B1}$ \\
            \hline $r_0$                            & $1$                      & $3$                      \\
            \hline $r_1$                            & $0$                      & $5$                      \\
            \hline
        \end{tabular}
    \end{minipage}

    \begin{minipage}{0.48\textwidth}
        \centering
        \begin{tabular}{|c|c|c|}
            \hline $u_B( \cdot,\theta_{B0}, \cdot)$ & $\theta_A = \theta_{A0}$ & $\theta_A = \theta_{A1}$ \\
            \hline $r_0$                            & $5$                      & $10$                     \\
            \hline $r_1$                            & $10$                     & $0$                      \\
            \hline
        \end{tabular}
    \end{minipage}
    \hspace{1em}
    \begin{minipage}{0.48\textwidth}
        \centering
        \begin{tabular}{|c|c|c|}
            \hline $u_B( \cdot,\theta_{B1}, \cdot)$ & $\theta_A = \theta_{A0}$ & $\theta_A = \theta_{A1}$ \\
            \hline $r_0$                            & $10$                     & $10$                     \\
            \hline $r_1$                            & $10$                     & $4$                      \\
            \hline
        \end{tabular}
    \end{minipage}

\end{table}

    According to the algorithm in Section \ref{sec:algorithm}, we firstly calculate
    $W_0$, then as shown in Figure \ref{fig:inf_round_welfare}, we denote
    $W_k\left(\frac{2}{5}, \frac{4}{5}\right)$ as $m_1^{(k)}$,
    $W_k\left(\frac{1}{2}, \frac{4}{5}\right)$ as $m_2^{(k)}$,
    $W_k\left(\frac{2}{3}, \frac{4}{5}\right)$ as $m_3^{(k)}$,
    $W_k\left(\frac{2}{5}, \frac{3}{5}\right)$ as $m_4^{(k)}$,
    $W_k\left(\frac{1}{2}, \frac{3}{5}\right)$ as $m_5^{(k)}$, and
    $W_k\left(\frac{2}{3}, \frac{3}{5}\right)$ as $m_6^{(k)}$.

    \begin{figure}[h]
        \centering
        \begin{tikzpicture}
            \begin{axis}[
                    xmin=-0.03,
                    xmax=1.15,
                    ymin=-0.03,
                    ymax=1.15,
                    axis lines = middle,
                    xlabel = {$q_B(\theta_B = \theta_{B0})$},
                    ylabel = {$q_A(\theta_A = \theta_{A0})$},
                    width = 0.7\textwidth,
                    height = 0.63\textwidth,
                    x label style={at={(axis description cs:0.5,-0.05)},anchor=north},
                    y label style={at={(axis description cs:-0.05,.5)},rotate=90,anchor=south},
                    xtick = {0, 0.4, 0.5, 2/3, 1},
                    xticklabels = {$0$, $\frac{2}{5}$,$\frac{1}{2}$, $\frac{2}{3}$, $1$},
                    ytick = {0, 0.6, 0.8, 1},
                    yticklabels = {$0$, $\frac{3}{5}$, $\frac{4}{5}$, $1$},
                    tick label style = {font=\footnotesize}
                ]

                \addplot[color = gray, dashed, line width = 0.3mm, domain = 0:1]{0.6};
                \addplot[color = gray, dashed, line width = 0.3mm, domain = 0:1]{0.8};
                \addplot[color = gray, dashed, line width = 0.3mm, domain = 0:1]{1};

                \addplot[color = gray, dashed, line width = 0.3mm] coordinates {(0.5,0)(0.5,1)};
                \addplot[color = gray, dashed, line width = 0.3mm] coordinates {(0.4,0)(0.4,1)};
                \addplot[color = gray, dashed, line width = 0.3mm] coordinates {(2/3,0)(2/3,1)};
                \addplot[color = gray, dashed, line width = 0.3mm] coordinates {(1,0)(1,1)};

                \filldraw[blue] (2/5,0) circle (2pt) node[anchor=south east]{$-\infty$};
                \filldraw[red] (2/5, 3/5) circle (2pt) node[anchor=south east]{$m_4^{(0)} =\frac{297}{25}$} node[anchor=north east]{$(q_B^0, q_A^0)$};
                \filldraw[blue] (2/5,4/5) circle (2pt) node[anchor=south east]{$m_1^{(0)} =\frac{351}{25}$};
                \filldraw[blue] (2/5,1) circle (2pt) node[anchor=south east]{$\frac{81}{5}$};

                \filldraw[blue] (0,0) circle (2pt) node[anchor=south west]{$-\infty$};
                \filldraw[blue] (1/2, 0) circle (2pt) node[anchor=south west]{$-\infty$};
                \filldraw[blue] (2/3,0) circle (2pt) node[anchor=south west]{$\frac{35}{3}$};
                \filldraw[blue] (1,0) circle (2pt) node[anchor=south west]{$11$};

                \filldraw[blue] (0,3/5) circle (2pt) node[anchor=south west]{$\frac{69}{5}$};
                \filldraw[blue] (1/2, 3/5) circle (2pt) node[anchor=south]{$\ \ m_5^{(0)} =\frac{57}{5}$};
                \filldraw[blue] (2/3,3/5) circle (2pt) node[anchor=south west]{$m_6^{(0)} =\frac{187}{15}$};
                \filldraw[blue] (1,3/5) circle (2pt) node[anchor=south west]{$\frac{58}{5}$};

                \filldraw[blue] (0,4/5) circle (2pt) node[anchor=south west]{$\frac{77}{5}$};
                \filldraw[blue] (1/2, 4/5) circle (2pt)
                node[anchor=south]{$\ \ m_2^{(0)} = \frac{137}{10}$};
                \filldraw[blue] (2/3,4/5) circle (2pt) node[anchor=south west]{$m_3^{(0)} =\frac{191}{15}$};
                \filldraw[blue] (1,4/5) circle (2pt) node[anchor=south west]{$\frac{59}{5}$};

                \filldraw[blue] (0,1) circle (2pt) node[anchor=south west]{$17$};
                \filldraw[blue] (1/2, 1) circle (2pt) node[anchor=south west]{$16$};
                \filldraw[blue] (2/3,1) circle (2pt) node[anchor=south west]{$-\infty$};
                \filldraw[blue] (1,1) circle (2pt) node[anchor=south west]{$-\infty$};

            \end{axis}
        \end{tikzpicture}

        \caption{Illustration of $W_0$ in the two-player game where the highest social welfare requires infinite rounds of communication to achieve.}
        \label{fig:inf_round_welfare}
    \end{figure}

    After calculations, we have discovered the following iterative pattern: For all
    $k>0$, When $k$ is even,
    \begin{align*}
        m_1^{(k+1)} = & m_1^{(k)},                                                       \\
        m_2^{(k+1)} = & m_2^{(k)},                                                       \\
        m_3^{(k+1)} = & \frac{2}{3}m_2^{(k)}+\frac{1}{3}W_k\left(1, \frac{4}{5}\right) , \\
        m_4^{(k+1)} = & \frac{3}{5}m_6^{(k)}+\frac{2}{5}W_k\left(0, \frac{3}{5}\right),  \\
        m_5^{(k+1)} = & \frac{3}{4}m_6^{(k)}+\frac{1}{4}W_k\left(0, \frac{3}{5}\right),  \\
        m_6^{(k+1)} = & m_6^{(k)}.
    \end{align*}
    when $k$ is odd,
    \begin{align*}
        m_1^{(k+1)} = & \frac{1}{2}m_4^{(k)}+\frac{1}{2}W_k\left(\frac{2}{5}, 1\right), \\
        m_2^{(k+1)} = & \frac{1}{2}m_5^{(k)}+\frac{1}{2}W_k\left(\frac{1}{2}, 1\right), \\
        m_3^{(k+1)} = & m_3^{(k)},                                                      \\
        m_4^{(k+1)} = & m_4^{(k)},                                                      \\
        m_5^{(k+1)} = & m_5^{(k)},                                                      \\
        m_6^{(k+1)} = & \frac{3}{4}m_3^{(k)}+\frac{1}{4}W_k\left(\frac{2}{3}, 0\right).
    \end{align*}
    And in all iterations, the values of other points remain unchanged: For all $k>0$, for all $(x,y)\in X^*\times Y^*-\{\left(\frac{2}{5}, \frac{4}{5}\right),\left(\frac{1}{2}, \frac{4}{5}\right), \left(\frac{2}{3}, \frac{4}{5}\right),\left(\frac{2}{5}, \frac{3}{5}\right), \left(\frac{1}{2}, \frac{3}{5}\right) ,\left(\frac{2}{3}, \frac{3}{5}\right)\}$,
    \begin{align*}
        W_{k+1}(x, y) = W_k(x,y).
    \end{align*}

    Since for even $k$,
    \begin{align*}
        m_4^{(k+1)} = \frac{3}{5}m_6^{(k)}+\frac{2}{5}W_k\left(0, \frac{3}{5}\right),\qquad
        m_5^{(k+1)} = \frac{3}{4}m_6^{(k)}+\frac{1}{4}W_k\left(0, \frac{3}{5}\right),
    \end{align*}
    and for odd $k$, $m_4$ and $m_5$ remain unchanged, thus we know that $W_k(\cdot,\frac{3}{5})$ is linear in $(0, \frac{1}{2})$ for all $k$.

    For all $k$, $W_k(0,\frac{3}{5}),W_k(0,\frac{4}{5}),W_k(0,1)$ remain unchanged,
    and their values are linear in $y$. And for all $k$,
    $W_k(0,1),W_k(\frac{2}{5},1),W_k(\frac{1}{2},1)$ remain unchanged, and their
    values are linear in $x$.

    For odd $k$,
    \begin{align*}
        m_1^{(k+1)} = \frac{1}{2}m_4^{(k)}+\frac{1}{2}W_k\left(\frac{2}{5}, 1\right),\qquad
        m_2^{(k+1)} = \frac{1}{2}m_5^{(k)}+\frac{1}{2}W_k\left(\frac{1}{2}, 1\right)
    \end{align*}
    So
    \begin{align*}
        m_1^{(k+1)} = & \frac{1}{2}m_4^{(k)}+\frac{1}{2}W_k\left(\frac{2}{5}, 1\right)                                                                                                                        \\
        =             & \frac{1}{2}\left(\frac{4}{5}m_5^{(k)}+\frac{1}{5}W_k\left(0,\frac{3}{5}\right)\right)+\frac{1}{2}\left(\frac{4}{5}W_k\left(\frac{1}{2},1\right)+\frac{1}{5}W_k\left(0,1\right)\right) \\
        =             & \frac{4}{5}m_2^{(k+1)}+\frac{1}{5}W_{k+1}\left(0,\frac{4}{5}\right).
    \end{align*}
    and for even $k$, $m_1$ and $m_2$ remain unchanged, making $W_k(\cdot,\frac{4}{5})$ linear in $(0, \frac{1}{2})$ for all $k$, so:
    \begin{align*}
        m_2^{(k)} = \frac{5}{4} m_1^{(k)}-\frac{1}{4}W_{k}\left(0,\frac{4}{5}\right)=\frac{5}{4} m_1^{(k)}-\frac{77}{20},\quad \forall k.
    \end{align*}

    For odd $k$,
    \begin{align*}
        m_4^{(k+4)} = & \frac{3}{5}m_6^{(k+3)}+\frac{138}{25},                                     \\
        m_6^{(k+3)} = & \frac{3}{4}m_3^{(k+2)}+\frac{35}{12},                                      \\
        m_3^{(k+2)} = & \frac{2}{3}m_2^{(k+1)}+\frac{59}{15}=\frac{5}{6}m_1^{(k+1)}+\frac{41}{30}, \\
        m_1^{(k+1)} = & \frac{1}{2}m_4^{(k)}+\frac{81}{10}.
    \end{align*}

    By recursively applying the above equations, we have for odd $k$: $$
        m_4^{(k+4)} = \frac{3}{16}m_4^{(k)}+\frac{4369}{400}. $$

    Since $m_4^{(1)}=13$, we get $$ W_{4k+1}(0.4,0.6)=m_4^{(4k+1)} =
        -\frac{144}{325}\left(\frac{3}{16}\right)^k +\frac{4369}{325}. $$

    And from this expression, we can see that the optimal welfare is never achieved
    by any finite value of $k$.
\end{proof}

\end{document}